\documentclass[10pt]{article}

\usepackage[T1]{fontenc}
\usepackage[english]{babel}
\usepackage{csquotes}          
\usepackage{lmodern}
\usepackage{microtype}         
\usepackage[multiple]{footmisc}

\usepackage[margin=2cm]{geometry}
\usepackage{setspace}

\usepackage{amsmath,amssymb,amsthm,mathtools}
\usepackage{bm}
\usepackage{extarrows}
\usepackage{dsfont}
\usepackage{tabto}
\usepackage{colonequals}
\usepackage{nicefrac}
\usepackage{tikz-cd}

\usepackage{graphicx}
\usepackage{booktabs}
\usepackage{tabularx}
\usepackage{multirow}
\usepackage[shortlabels]{enumitem}

\usepackage{xcolor}
\usepackage[colorlinks,allcolors=blue]{hyperref}
\usepackage{cleveref}

\usepackage[backend=biber,style=numeric,sorting=none]{biblatex}
\usepackage[ruled,vlined]{algorithm2e}
\SetAlgoNlRelativeSize{-1}

\usepackage[
  round-mode          = figures , 
  round-precision     = 3       , 
  scientific-notation = true      
]{siunitx}

\newcommand{\R}{\mathbb{R}}
\newcommand{\C}{\mathbb{C}}

\newcommand{\iu}{\mrm{i}\mkern1mu}

\newcommand{\tran}{\mathsf{T}}
\newcommand{\herm}{\mathsf{H}}
\newcommand{\frob}{\mathsf{F}}

\DeclareMathOperator*{\argmax}{argmax}
\DeclareMathOperator*{\argmin}{argmin}

\DeclarePairedDelimiter{\abs}{\lvert}{\rvert}
\DeclarePairedDelimiter{\norm}{\lVert}{\rVert}
\DeclarePairedDelimiterX{\inner}[2]{\langle}{\rangle}{{#1},{#2}}
\DeclarePairedDelimiter{\paren}{(}{)}
\DeclarePairedDelimiter{\sbrack}{[}{]}
\DeclarePairedDelimiter{\cbrack}{\lbrace}{\rbrace}

\allowdisplaybreaks
\setlist[itemize]{leftmargin=*,itemsep=0pt,labelsep=10pt}
\setlist[enumerate]{leftmargin=*,itemsep=0pt,labelsep=10pt}
\numberwithin{equation}{section}

\newtheorem{theorem}{Theorem}[section]
\newtheorem{lemma}[theorem]{Lemma}
\newtheorem{proposition}[theorem]{Proposition}

\crefname{equation}{Eq.}{Eqs.}
\Crefname{equation}{Eq.}{Eqs.}
\crefname{figure}{Figure}{Figures}
\crefname{table}{Table}{Tables}
\crefname{section}{Section}{Sections}
\crefname{lemma}{Lemma}{Lemmata}
\Crefname{algocf}{Algorithm}{Algorithms}

\newcommand{\Cov}{\boldsymbol{\mrm{C}}}
\newcommand{\Cove}{\mrm{C}}
\newcommand{\BSigma}{\boldsymbol{\Sigma}}
\newcommand{\btheta}{\boldsymbol{\theta}}
\newcommand{\bvtheta}{\boldsymbol{\vartheta}}
\newcommand{\Tau}{\mrm{T}}
\newcommand{\BTau}{\boldsymbol{\mrm{T}}}
\newcommand{\BPhi}{\boldsymbol{\Phi}}
\newcommand{\mrm}[1]{\mathrm{#1}}
\newcommand{\Bmrm}[1]{\boldsymbol{\mrm{#1}}}

\begin{document}

\title{Heuristic Quality Coefficients for Interferometric Phase Linking}

\author{Magnus Heimpel $^{1,}$*, Irena Hajnsek $^{1,2}$ and Othmar Frey $^{1,3}$\\\vspace{\baselineskip}\\
$^{1}$ \quad Institute of Environmental Engineering, ETH Zurich, Switzerland\\
$^{2}$ \quad Microwaves and Radar Institute, German Aerospace
Center (DLR), Weßling, Germany\\
$^{3}$ \quad Gamma Remote Sensing AG, Gümligen, Switzerland\\
$^\ast$ \quad Correspondence: mheimpel@ethz.ch
}
\maketitle

\begin{abstract}
In multitemporal InSAR, phase linking (PL) refers to the estimation of a single-reference interferometric phase history for distributed scatterers (DS) from the information contained in the sample coherence matrix. Because the phase information in this matrix is typically inconsistent, DS processing needs practical reliability indicators to decide whether a pixel's PL estimate is sufficiently supported by the data for subsequent deformation analysis. For maximum-likelihood estimation, uncertainty can be quantified via Fisher-information-based covariance estimates, but no analogous, generally applicable uncertainty quantification is available for the broad range of non-ML methods. We propose three heuristic quality coefficients within a unified mathematical framework that covers common PL methods: (1) a method-specific goodness-of-fit coefficient that normalizes the achieved PL objective between a method-consistent upper bound and an empirically modeled noise floor level; (2) a closure phase coefficient computed from the sample coherence matrix in advance; and (3) an ambiguity coefficient that compares the obtained PL estimate with the best alternative in its orthogonal complement in the solution space. All coefficients are normalized to the interval $[0,1]$, where 1 indicates maximum reliability and 0 matches the behavior expected under pure noise. Simulations under exponential and seasonal decorrelation models show that the goodness-of-fit coefficient tracks the normalized absolute phase error most consistently, whereas the closure phase coefficient provides an a priori indicator for pre-screening. Experiments on a TerraSAR-X stack over Visp, Switzerland, reveal plausible spatial patterns across urban and vegetated areas and show that the ambiguity coefficient provides complementary information, especially in regions with temporally varying scattering mechanisms.\vspace{\baselineskip}

\end{abstract}

\section{Introduction}\label{sec:Introduction}

Interferometric Synthetic Aperture Radar (InSAR) allows the retrieval of surface motion based on measured phase differences between repeated SAR acquisitions. Its relevance lies in its unique ability to provide spatially dense measurements over large areas, with the potential to achieve up to millimeter precision in the line-of-sight direction under favorable conditions. The precision of the phase measurements depends strongly on the interferometric coherence between acquisitions. Decorrelation of the radar targets introduces uncertainty in the interferometric phase, which propagates to the estimated surface deformation. Point-like scatterers with high phase stability that dominate the backscatter within a resolution cell, also known as persistent scatterers (PS), serve as suitable targets to retrieve surface motion~\cite{ferretti2001permanent,werner2003interferometric,mora2003linear,hooper2004new}. Tomographic approaches extend the retrieval to cases in which more than one such target is present in a range-azimuth resolution cell~\cite{lombardini2005differential,fornaro2009four,zhu2010very,siddique2016single,siddique2018sar}. However, in regions with vegetation, especially at shorter wavelengths, stable point-like scatterers are typically sparse. In these cases, statistical analysis of distributed scatterers (DS)~\cite{even2018insar}, which represent the combined return of multiple scatterers that are similar and much smaller than a resolution cell, may be necessary to achieve sufficient coverage across the scene.

DS targets tend to be more prone to decorrelation, and their interferometric phases are generally inconsistent across different acquisition pairs. Therefore, strategies are needed to extract usable phase information from them. A common approach is to use phase linking (PL) methods, which estimate a single-reference interferometric phase history from the phases of all available acquisition pairs. Over the past two decades, a variety of such techniques have been developed. Despite methodological differences, most approaches can be viewed as defining an objective that measures the alignment between a candidate phase history and the observed phases in the sample coherence matrix, and then solving the resulting optimization problem.

Among the most established methods are the small baseline subset (SBAS) method~\cite{berardino2002new,schmidt2003time}; phase triangulation (PT) approaches such as those proposed by Monti-Guarnieri and Tebaldini~\cite{montiguarnieri2008exploitation}, Ferretti et al.~\cite{ferretti2011new}, and Cao et al.~\cite{cao2015mathematical}; and methods based on eigendecomposition (ED) such as CAESAR~\cite{fornaro2015caesar} and EMI~\cite{ansari2018efficient}. More recently, Bai et al.~\cite{bai2023lamie} proposed LaMIE for phase estimation from large-dimensional coherence matrices. Sequential estimators~\cite{ansari2017sequential,mirzaee2023non} extend these approaches to large-scale datasets by sequentially processing subsets of the data stack, which can improve scalability and allow efficient incorporation of new acquisitions.

PL is typically embedded in a longer DS-InSAR processing chain, where only pixels deemed reliable are carried forward to deformation analysis and geophysical interpretation. This makes reliability assessment a central practical step: for a chosen PL method, one needs a scalar indicator that reflects whether the resulting estimate at a pixel is sufficiently supported by the data to warrant inclusion. Several reliability indicators have been proposed to support such decisions. For maximum-likelihood PL, Fisher-information-based covariance estimates can provide uncertainty measures~\cite{zwieback2022reliable}. However, no analogous, generally applicable uncertainty quantification is available for non-ML methods. Ferretti et al.~\cite{ferretti2011new} proposed a phase residual metric for filtering unreliable points in the SqueeSAR framework. Samiei-Esfahany et al.~\cite{samieiesfahany2016phase} incorporated the same metric into their integer least squares estimator and Even and Schulz~\cite{even2018insar} discussed refining it with coherence-based weighting. Ansari et al.~\cite{ansari2018efficient} proposed an indicator derived from the eigenstructure of the coherence matrix in the EMI framework. While effective within their respective settings, these indicators are typically formulated only for the specific estimators in which they were introduced, and their scale depends on the number of acquisitions and the sample size, which complicates threshold selection across different scenarios.

This work develops a unified framework for commonly used PL methods; the key idea is to view PL as a (weighted) rank-one approximation problem on the unit-phasor matrix derived from the sample coherence. Within this framework, we formalize and generalize the goodness-of-fit principle underlying established indicators~\cite{ferretti2011new,even2018insar,ansari2018efficient} and normalize it to improve interpretability. We also propose two additional coefficients targeting complementary aspects of reliability. The proposed coefficients are intended as reliability indicators for a PL result conditional on a chosen estimator; downstream deformation-accuracy validation is beyond the scope of this work. Concretely, our contributions are:
\begin{enumerate}
  \item In \Cref{sec:basis}, we extend the framework of Cao et al.~\cite{cao2015mathematical} and show that a broad range of PL methods can be written in a common optimization form. This clarifies relationships between established estimators and provides a consistent categorization by feasible set (PT vs.\ ED) and weights. The perspective is related to the COFI-PL formulation of Vu et al.~\cite{vu2025covariance}, which interprets PL as a projection onto phase-consistent matrices, while our presentation emphasizes a complementary rank-one approximation viewpoint and accommodates the weighting schemes considered here.
    \item In \Cref{sec:coefficients}, we introduce three normalized scalar quality coefficients: \begin{itemize}
        \item a \emph{goodness-of-fit coefficient}, \Cref{eq:gofcoef}, that normalizes each method's objective using upper and lower bounds together with an empirically modeled noise floor;
        \item a method-independent \emph{closure phase coefficient}, \Cref{eq:clippedcpcoef}, that quantifies internal phase consistency and can be computed directly from the sample coherence matrix;
        \item an \emph{ambiguity coefficient}, \Cref{eq:ambcoef}, that compares the obtained PL estimate to the best alternative in its orthogonal complement in the solution space to assess how distinguished the solution is.
\end{itemize}
    \item In \Cref{sec:experiments}, we evaluate these quality coefficients in two simulation scenarios (exponential and seasonal decorrelation) and on a TerraSAR-X stack over Visp, Switzerland. The experiments assess how well the coefficients track the normalized absolute error in simulation and how they behave spatially for real data across urban and strongly decorrelated vegetated areas.
    \item In Appendix~\ref{sec:solvers}, we discuss practical solvers for PT and propose a robust nonlinear conjugate-gradient approach that is applicable across the PT objectives considered here and is compatible with the constrained optimization required for computing the ambiguity coefficient.
\end{enumerate}

\section{Theoretical Basis}\label{sec:basis}

In the following, we present a mathematical framework that encompasses a range of PL methods and enables us to formulate the quality coefficients in \Cref{sec:coefficients} in a unified manner. For overviews of DS-InSAR and PL, we refer to Even and Schulz~\cite{even2018insar} and Minh and Tebaldini~\cite{minh2023interferometric}, respectively.

\subsection{Coherence Matrix}\label{sec:cohmat}

Consider a stack of $N$ coregistered SLC images. We denote by $\Omega\subset\C^N$ the set of pixels belonging to a particular DS, such that each element of $\Omega$ is an $N$-dimensional vector representing the multitemporal SAR data at the corresponding location.\footnote{Identifying these sets of statistically homogeneous pixels (SHPs) is a nontrivial task in its own right; see~\cite{ferretti2011new,jiang2015fast,dong2018unified}.} We adopt the following simple model\footnote{Some works consider more general settings; see~\cite[p.~4 ff.]{even2018insar}.}: the vectors in $\Omega$ are samples from a zero-mean, circularly symmetric complex normal distribution with covariance matrix $\Cov\in\C^{N\times N}$, as expressed in \Cref{eq:ccg}.
\begin{equation}\label{eq:ccg}
     \forall\,\boldsymbol{\omega}\in\Omega:\quad\boldsymbol{\omega}=\begin{bmatrix} \omega_1 & \cdots & \omega_N \end{bmatrix}^\tran\sim \mathcal{CN}(0, \Cov, 0).
 \end{equation}
 
Moreover, we assume that backscatter is present in every acquisition, so that $0<\Cove_{ii}\in\R$ holds for all indices $i$. The coherence matrix $\BSigma\in\C^{N\times N}$ and the sample coherence matrix $\BTau\in\C^{N\times N}$ are defined (entrywise) by
\begin{equation}\label{eq:coherencedef}
     \Sigma_{ij} \coloneqq \frac{\Cove_{ij}}{\sqrt{\Cove_{ii}}\, \sqrt{\Cove_{jj}}}\qquad\text{and}\qquad\Tau_{ij} \coloneqq \frac{\sum_{\boldsymbol{\omega}\in\Omega}\omega_i\,\overline{\omega_j}}{\sqrt{\sum_{\boldsymbol{\omega}\in\Omega}\vert\omega_i\vert^2}\, \sqrt{\sum_{\boldsymbol{\omega}\in\Omega}\vert\omega_j\vert^2}}.
 \end{equation}
 For simplicity, we restrict attention to $\BTau$ as an estimator\footnote{Estimating $\BSigma$ via $\BTau$ implicitly assumes that $\Omega$ is much smaller than any interferometric fringe at the same location; cf. Zebker and Chen~\cite{zebker2005accurate}. In practice, this may require removing the topographic phase component and any strong phase ramps from other sources, if present.} of $\BSigma$, although many improvements have been proposed~\cite{lopezmartinez2007coherence, wang2016robust, vu2023robust,adam2025improved}.
 
Let $\circ$ denote the Hadamard product, and define the $N$-torus by
\begin{equation}
\mathbb{T}^N\coloneqq\cbrack{\boldsymbol{v}\in\C^N\,\mid\,\forall\,i:\abs{v_i}=1}.
\end{equation}
We adopt the fundamental assumption that the sets
\begin{align}\label{eq:deflinkage}
    \Theta_{\R}&\coloneqq \cbrack*{ \btheta\in{(-\pi,\pi]}^N \,\mid\, \forall\,i,j:\Sigma_{ij}=\abs{\Sigma_{ij}}\exp\!\big(\iu\big(\theta_i-\theta_j\big)\big)} \quad\text{and}\\ \Theta_{\C}&\coloneqq\cbrack*{\bvtheta\in\mathbb{T}^N \,\mid\, \BSigma = \abs{\BSigma}\circ\bvtheta\bvtheta^\herm},
\end{align}
which we refer to as the (real and complex) linkage of $\BSigma$, are nonempty~\cites[Eq.~(5)]{montiguarnieri2008exploitation}[Eq.~(5)]{ferretti2011new}[Eq.~(2)]{vu2023robust}. The inverse entrywise maps $\bvtheta=\exp(\iu\btheta)$ and $\btheta=\angle\bvtheta$ establish a bijection between them. We refer to elements of these sets as the (real and complex) phase histories of $\BSigma$. By definition, $\Theta_{\C}$ (and $\Theta_{\R}$, analogously) is invariant under global phase shifts of the form $\bvtheta\mapsto\exp(\iu\theta_{\mrm{shift}})\bvtheta$. Furthermore, we call $\Theta_{\R}$ and $\Theta_{\C}$ regular if any two phase histories differ only by a global phase shift, i.e.,
\begin{equation}\label{eq:conditionregular}
    \forall\,\bvtheta_1,\bvtheta_2\in\Theta_{\C}:\; \paren[\Big]{\exists\,\theta_{\mrm{shift}}\in(-\pi,\pi]:\; \bvtheta_2=\exp(\iu\theta_{\mrm{shift}})\bvtheta_1}.
\end{equation}
Condition~\eqref{eq:conditionregular} is equivalent to $\operatorname{dim}(\operatorname{span}(\Theta_{\C}))=1$. For example, $\Theta_{\C}$ is nonregular when $\BSigma$ has a block structure, corresponding to a total loss of correlation across the acquisition stack. Now consider estimating InSAR parameters (e.g.,~surface displacement) relative to a chosen reference acquisition with index $p$. To this end, we seek a phase history $\btheta\in\Theta_{\R}$ satisfying $\theta_p=0$. If the linkage is regular, then such a phase history is unique and given by 
\begin{equation}\label{eq:uniquehistory}
    \btheta=\angle\paren*{\overline{\vartheta_p}\,\bvtheta}\;\text{ for arbitrary }\;\bvtheta\in\Theta_{\C}.
\end{equation}

Unlike $\BSigma$, the matrix $\BTau$ typically does not admit a decomposition of the form $\BTau = \abs{\BTau}\circ\bvtheta\bvtheta^\herm$. Otherwise, PL would be trivial. Instead, we define $\BPhi\in\C^{N\times N}$ by
\begin{equation}\label{eq:phasematdef}
\Phi_{ij}\coloneqq\exp\!\big(\iu\phi_{ij}\big),\quad\text{with}\quad\phi_{ij}\coloneqq
\begin{cases}\angle \Tau_{ij},&\abs{\Tau_{ij}}>0\\0,&\abs{\Tau_{ij}}=0\end{cases},
\end{equation}
so that $\BTau=\abs{\BTau}\circ\BPhi$. For a given triplet of indices $(i,j,k)$, we define the closure phase~\cite{dezan2015phase,zheng2022closure} by 
\begin{equation}\label{eq:defclosurephase}
\phi^\Delta_{ijk}\coloneqq\angle(\Phi_{ij}\Phi_{jk}\Phi_{ki})=\angle(\Phi_{ij}\Phi_{jk}/\Phi_{ik}).
\end{equation}
Apart from vanishing entries of $\BTau$, nonzero closure phases correspond either to estimation errors or to scattering processes that are inconsistent with the existence of a phase history~\cite{dezan2018vegetation,zwieback2016statistical}. Moreover, the following holds (\Cref{proposition:rankone}):
\begin{align}\label{eq:implications}
    &\operatorname{rank}\big(\BTau\big) = 1 \; &&\Leftrightarrow
\quad \exists\,\bvtheta\in\mathbb{T}^N:\;\BTau = \bvtheta \bvtheta^{\herm} \; &&\Leftrightarrow
\quad \forall\,i,j:\;\abs{\Tau_{ij}}=1 \nonumber\\ &\;\Downarrow && &&
\\ &\operatorname{rank}\big(\BPhi\big) = 1 \; &&\Leftrightarrow
\quad \exists\,\bvtheta\in\mathbb{T}^N:\;\BPhi = \bvtheta \bvtheta^{\herm} \; &&\Leftrightarrow
\quad \forall\,i,j,k:\;\phi^\Delta_{ijk}=0.\nonumber
\end{align}
This motivates low-rank approximations of $\BPhi$ as a strategy of PL.

\subsection{Phase Linking}\label{sec:phaselinking}
In this work, phase linking refers to the estimation of the linkage $\Theta_{\R}$ (or $\Theta_{\C}$, cf.~\Cref{eq:deflinkage}) of the unknown coherence matrix $\BSigma$, given the sample coherence matrix $\BTau=\abs{\BTau}\circ\BPhi$ and assuming that the linkage is nonempty and regular. We restrict attention to estimators of the form
\begin{equation}\label{eq:defPLmethods}
    \hat{\Theta}_{\R} \coloneqq \cbrack*{\angle\bvtheta \,\mid\,\bvtheta\in\argmax_{\boldsymbol{v}\in\mathcal{X}} \inner*{\boldsymbol{v}\boldsymbol{v}^\herm}{\Bmrm{W}\circ\BPhi}_\frob},
\end{equation}
where
\begin{itemize}
    \item $\Bmrm{W} \in \R^{N \times N}$ is a symmetric matrix,
    \item $\inner*{\Bmrm{A}}{\Bmrm{B}}_\frob = \sum_{i,j}\overline{\mrm{A}_{ij}}\,\mrm{B}_{ij}$ denotes the Frobenius inner product, and
    \item $\mathcal{X}$ is either $\mathbb{T}^N$ or $\mathbb{U}^N\coloneqq\cbrack*{\boldsymbol{v}\in\C^{N}\,\mid\,\norm{\boldsymbol{v}}_2=1\;\wedge\;\forall\,i:v_i\neq 0}$.
\end{itemize}
All methods mentioned in \Cref{sec:Introduction} are of this form (see~\cite{cao2015mathematical}, \Cref{lemma:realvalpt,proposition:lamie}, and \Cref{tab:methodsweightstable}), except for SBAS. However, Cao et al.~\cite[Sec.~II.E.]{cao2015mathematical} show that SBAS can be interpreted as a phase-unwrapped\footnote{We refer to~\cite{yu2019phase} for a description of interferometric phase unwrapping.} analogue of equal-weighted PT, which is included here. Intuitively, the rank-one approximation in \Cref{eq:defPLmethods} follows the same principle as approximating one vector by another in a subset of $\R^N$ by maximizing their Euclidean inner product. Alternatively, most methods can be formulated as projections that minimize a distance function specific to each method~\cite{vu2025covariance}. In the following, we present a simple categorization based on the choice of $\mathcal{X}$ and $\Bmrm{W}$.

\subsubsection{Phase Triangulation and Eigendecomposition}

We define PT and ED methods as the classes of estimators for which $\mathcal{X}$ in \Cref{eq:defPLmethods} equals $\mathbb{T}^N$ and $\mathbb{U}^N$, respectively. As shown in \Cref{lemma:realvalpt}, PT is equivalent to the following unconstrained nonlinear optimization problem:
\begin{equation}\label{eq:defPTreal}
    \hat{\Theta}_{\mrm{PT}} = \argmax_{\boldsymbol{\theta}\in{(-\pi,\pi]}^N} \left\{\sum_{i=1}^N \sum_{j=i+1}^N\mrm{W}_{ij} \cos\big(\phi_{ij} - \big(\theta_i-\theta_j\big)\big)\right\}.
\end{equation}

For $\Bmrm{A}\in\C^{N\times N}$, let $\lambda_{\max}\paren{\Bmrm{A}}$ and $E_{\lambda_{\max}}(\Bmrm{A})$ denote the largest eigenvalue of $\Bmrm{A}$ and the corresponding eigenspace, respectively. By the Rayleigh--Ritz principle and \Cref{lemma:realvalpt},
\begin{equation}
    \inner*{\boldsymbol{v}\boldsymbol{v}^\herm}{\Bmrm{W}\circ\BPhi}_\frob = \boldsymbol{v}^\herm \paren*{\Bmrm{W}\circ\BPhi} \boldsymbol{v} \leq \lambda_{\max}\paren*{\Bmrm{W}\circ\BPhi}
\end{equation}
for all $\boldsymbol{v}\in\mathbb{U}^N$. Consequently, ED methods can be expressed as
\begin{equation}\label{eq:defEDmethods}
    \hat{\Theta}_{\mrm{ED}} \coloneqq \cbrack*{\angle\boldsymbol{v}\,\mid\,\boldsymbol{v}\in E_{\lambda_{\max}}\paren[\big]{\Bmrm{W}\circ\BPhi}\cap\mathbb{U}^N}.
\end{equation}
The constraint $\forall\,i:v_i\neq 0$ in the definition of $\mathbb{U}^N$ ensures that $\angle\boldsymbol{v}$ in \Cref{eq:defPLmethods,eq:defEDmethods} is well-defined. Note that $\hat{\Theta}_{\mrm{ED}}$ may be empty if $E_{\lambda_{\max}}\paren*{\Bmrm{W}\circ\BPhi}$ and $\mathbb{U}^N$ do not intersect.\footnote{An example for $N=3$ is $\Bmrm{W}^{\mrm{CW}}\circ\BPhi=[1, 0.3, 0.3], [0.3, 1, -0.7], [0.3, -0.7, 1]$.} Vu et al.~\cite{vu2025covariance} describe ED methods as relaxations of PT methods, typically at the cost of reduced accuracy.

\subsubsection{Weights}\label{sec:sectionweights}

We categorize the methods in  \Cref{eq:defPLmethods} by the choice of the weight matrix $\Bmrm{W}$. Common choices include
\begin{itemize}
\item maximum-likelihood (ML): \tabto{5cm} $\Bmrm{W}^{\mrm{ML}}\coloneqq-\abs{\BTau}^{-1}\circ\abs{\BTau}$ 
\item coherence-weighted (CW): \tabto{5cm} $\Bmrm{W}^{\mrm{CW}}\coloneqq\abs{\BTau}$
\item equal-weighted (EW): \tabto{5cm} $\Bmrm{W}^{\mrm{EW}}\coloneqq\mathds{1}_{N}$.
\end{itemize}
Note that $\Bmrm{W}^{\mrm{ML}}$ may contain negative entries and therefore is not a weight matrix in the usual sense. If $\abs{\BTau}$ follows an AR(1) structure, then \Cref{prop:ar1_inverse} implies that the induced ML weights are nonnegative.

One may also exclude certain interferometric pairs from the PL process by applying a mask via $\Bmrm{W}_{\Bmrm{M}} \coloneqq \Bmrm{M} \circ \Bmrm{W}$, such as
\begin{equation}
    \mrm{M}_{ij}=
\begin{cases}1, & \abs*{i - j}\leq n\\ 0, & \abs*{i - j} > n
\end{cases}\,,
\quad 1\leq n \leq N \qquad\text{and}\qquad \mrm{M}_{ij}=
\begin{cases}1, & \abs{\Tau_{ij}} \geq \tau \\ 0, & \abs{\Tau_{ij}} < \tau
\end{cases}\,,
\quad \tau\in(0, 1].
\end{equation}
The first example regularizes $\Bmrm{W}\circ\BPhi$ via tapering (cf. Vu et al.~\cite[Eq.~(23)]{vu2025covariance}). The second example suppresses acquisition pairs with low coherence. This is particularly important in the EW case, since the estimation of $\Phi_{ij}$ becomes unstable as $\abs{\Tau_{ij}}$ approaches zero (and when $\Tau_{ij}=0$, $\Phi_{ij}$ is set to 1 by convention in \Cref{eq:phasematdef}). \Cref{tab:methodsweightstable} summarizes how various well-known PL methods are categorized within our framework.

\begin{table}[ht]
    \centering
    \renewcommand{\arraystretch}{1.2}
    \begin{tabular}{@{} l l l l @{}}
        \toprule
        \textbf{Name} & \textbf{References} & \textbf{Weights} & \textbf{Class} \\
        \midrule
        
        Phase Linking
        & \cite{montiguarnieri2008exploitation,ferretti2011new}
        & $\Bmrm{W}^{\mrm{ML}}$
        & PT \\
        
        EMI
        & \cite{ansari2018efficient}
        & $\Bmrm{W}^{\mrm{ML}}$
        & ED \\
        
        PTCM
        & \cite{cao2015mathematical,even2018insar}
        & $\Bmrm{W}^{\mrm{CW}}$
        & PT \\
        
        CAESAR
        & \cite{fornaro2015caesar}
        & $\Bmrm{W}^{\mrm{CW}}$
        & ED \\
        
        Angular (Phase) Synchronization (PT)
        & \cite{singer2011angular,cao2015mathematical}
        & $\Bmrm{W}^{\mrm{EW}}$
        & PT \\
        
        Angular (Phase) Synchronization (ED)
        & \cite{singer2011angular,boumal2016nonconvex}
        & $\Bmrm{W}^{\mrm{EW}}$
        & ED \\
        
        LaMIE
        & \cite{bai2023lamie}
        & $\abs{\BTau}\circ\abs{\BTau}$
        & PT \\
        
        SBAS
        & \cite{berardino2002new,schmidt2003time}
        & $\Bmrm{M}\circ \Bmrm{W}^{\mrm{EW}}$
        & - \\
        \bottomrule
    \end{tabular}
    \caption{Overview of PL methods and their categorization within our framework. As discussed in~\cite[Sec.~II.E.]{cao2015mathematical}, SBAS can be interpreted as a phase-unwrapped analogue of EW-PT with a mask that selects short-baseline acquisition pairs.}
    \label{tab:methodsweightstable}
\end{table}

\section{Quality Coefficients}\label{sec:coefficients}

The primary objective of this work is to develop scalar indicators of reliability, hereafter referred to as quality coefficients, for the PL estimates produced by the algorithms introduced in \Cref{sec:basis}. We assume that $\Theta_{\R}$ is nonempty and regular, and fix a representative $\btheta\in\Theta_{\R}$ via \Cref{eq:uniquehistory} for a chosen reference index $p$. Let $\hat{\btheta}$ denote an estimate of $\btheta$ obtained through PL and referenced to the same index. We define the absolute (angular) error of~$\hat{\btheta}$ as 
\begin{equation}\label{eq:defae}
    \operatorname{AE}(\hat{\btheta}) \coloneqq \norm[\big]{\varepsilon(\hat{\btheta})}_2, \quad \text{where} \quad \varepsilon(\hat{\btheta}) \coloneqq \angle\exp\!\big(\iu\big(\hat{\btheta}-\btheta\big)\big).
\end{equation}

Ideally, one would estimate the expected $\operatorname{AE}(\hat{\btheta})$ from the observed data. To the best of our knowledge, no such estimators exist for the methods discussed in \Cref{sec:basis}. For ML methods, uncertainty quantification can instead be based on the covariance of the PL estimate derived from the Fisher information~\cite{zwieback2022reliable}. In particular, the absolute error and the covariance satisfy
\begin{equation}\label{eq:msecovbias}
    \mathbb{E}\sbrack*{{\operatorname{AE}(\hat{\btheta})}^2} = \operatorname{tr}\paren[\Big]{\operatorname{Cov}\paren[\big]{\varepsilon(\hat{\btheta})}} + \norm[\Big]{\mathbb{E}\big(\varepsilon(\hat{\btheta})\big)}_2^2,
\end{equation}
provided that the second moments exist, where the expectation is taken under the data-generating distribution parameterized by $\BSigma$ (i.e.,~over the randomness of the observations, and thus of $\BTau$ and $\hat{\btheta}$, conditional on $\BSigma$). For the other non-ML methods in \Cref{sec:basis}, analogous uncertainty estimators are not available. Moreover, bootstrap-based strategies are computationally expensive and may be unreliable at finite sample sizes.

While ML methods are asymptotically optimal under their model assumptions, for smaller sample sizes and for departures from these assumptions, non-ML methods can perform better~\cite[p.~23]{even2018insar}. This motivates reliability indicators for non-ML phase linking estimates to support informed DS pixel selection. Indeed, several works propose such indicators~\cites[Eq.~(16)]{ferretti2011new}[Eq.~(23)]{even2018insar}[Eq.~(18)]{ansari2018efficient}. Here, we formalize and generalize the heuristic principles underlying these indicators and propose alternatives.

We further define the normalized absolute error of~$\hat{\btheta}$ as
\begin{equation}\label{eq:defnae}
    \operatorname{nAE}(\hat{\btheta})\coloneqq\sqrt{\frac{3}{N\pi^2}}\operatorname{AE}(\hat{\btheta}),
\end{equation}
where the scaling is chosen such that $\mathbb{E}[\operatorname{nAE}(\hat{\btheta})]=1$ when the componentwise phase errors are i.i.d.\ uniform on $(-\pi,\pi]$. Nominally, we seek quality coefficients $\gamma\in\sbrack{0,1}$ with the following behavior:\pagebreak
\begin{align}
    &\BTau = \BSigma\;\wedge\,\BSigma\neq\Bmrm{I}_N\quad\Rightarrow\quad\gamma=1,\label{eq:bestcase}\\
    &\BSigma = \Bmrm{I}_N\quad\Rightarrow\quad\gamma=0,\label{eq:worstcase}\\
    &\operatorname{corr}\paren*{1 - \gamma,\;\operatorname{nAE}(\hat{\btheta})}>0,\label{eq:corrnae}
\end{align}
where the correlation is taken with respect to the distribution of $\BSigma$ across scenarios, together with the conditional data-generating distribution given $\BSigma$.

\Cref{eq:bestcase} states that, in the absence of coherence matrix estimation errors, the coefficient should be one. Assuming that $\Theta_{\R}$ is nonempty and regular, the phase history is then identifiable (up to the chosen reference) when $\BTau=\BSigma$. Conversely, \Cref{eq:worstcase} states that when the data is indistinguishable from noise, the coefficient should be zero. Since coherence matrix estimation becomes increasingly erroneous as $\BSigma$ approaches $\Bmrm{I}_N$, these statements do not contradict each other.

In the following, we propose three coefficients designed to exhibit the behavior above; we assess \Cref{eq:corrnae} empirically via simulations. Throughout, we interpret $\gamma$ as a reliability indicator for a fixed PL method. We do not treat differences in $\gamma$ across phase linking methods as a principled basis for method ranking. Using reliability indicators for method selection may be possible in principle, but it would require a separate evaluation and is beyond the scope of this work.

\subsection{Goodness-of-Fit Coefficient}\label{sec:gof}
In \Cref{sec:basis}, we interpreted PL methods as (weighted) rank-one approximations of the unit phasor matrix $\BPhi$. For a regular nonempty linkage (\Cref{eq:deflinkage}) in the absence of coherence matrix estimation errors (i.e.,~$\BTau=\BSigma$), it follows that $\BPhi$ has rank one itself, so both the approximation and the recovered phase histories are exact. This motivates measuring the quality of PL by how well $\BPhi$ is approximated.

A natural choice is to divide the achieved objective value by an upper bound on that objective; we refer to such a normalized quantity as a goodness-of-fit indicator. For instance, Ferretti et al.~\cite[Eq.~(16)]{ferretti2011new} introduced the indicator
\begin{equation}\label{eq:ferretticoef}
    \gamma_{\mrm{PTA}} \coloneqq \frac{2}{N^2-N}\sum_{i=1}^N\sum_{j=i+1}^N\cos\big(\phi_{ij} - (\hat{\theta}_i-\hat{\theta}_j)\big),
\end{equation}
which is the normalized objective in \Cref{eq:defPTreal} for $\Bmrm{W}^{\mrm{EW}}$. They retain pixels for further processing only when this coefficient exceeds a prescribed threshold. Even and Schulz~\cite[Eq.~(23)]{even2018insar} refine this measure by incorporating coherence magnitudes as weights:
\begin{equation}
    \gamma_{\mrm{PTA}w}\coloneqq\paren*{\sum_{i=1}^N\sum_{j=i+1}^N\abs{\Tau_{ij}}}^{-1}\sum_{i=1}^N\sum_{j=i+1}^N\abs{\Tau_{ij}}\cos\big(\phi_{ij} - (\hat{\theta}_i-\hat{\theta}_j)\big),
\end{equation}
which is the same normalized objective for $\Bmrm{W}^{\mrm{CW}}$.
For EMI, Ansari et al.~\cite[Eq.~(18)]{ansari2018efficient} propose
\begin{equation}
    \lambda \coloneqq \min_{\boldsymbol{v}\in\C^N} \left\{ \dfrac{-1}{\norm*{\boldsymbol{v}}_2^2} \sum_{i=1}^N\sum_{j=i+1}^N \Bmrm{W}^{\mrm{ML}}_{ij} \abs{v_i}\abs{v_j}\cos\big(\phi_{ij} - (\angle v_i-\angle v_j)\big)\right\},
\end{equation}
which is consistent with the ML-weighted objective in \Cref{eq:defEDmethods}.

As noted by Cao et al.~\cite{cao2015mathematical}, the coefficient in \Cref{eq:ferretticoef} corresponds to $\Bmrm{W}^{\mrm{EW}}$, whereas Ferretti et al.~\cite{ferretti2011new} actually use $\Bmrm{W}^{\mrm{ML}}$ in the corresponding PT estimator. In this work, we instead assess the goodness-of-fit of a PL method using the objective it optimizes, together with its associated weights. This choice is still not unique, because any strictly monotone transformation of the objective yields an equivalent optimization problem with the same maximizers. We therefore restrict attention to the natural objective in \Cref{eq:defPTreal} for PT methods and to $\lambda_{\max}(\Bmrm{W}\circ\BPhi)$ for ED methods.

While the above examples relate an achieved objective value to an upper bound, they do not compare it to a lower bound. Moreover, for similar decorrelation scenarios, the achieved objective typically decreases when either the number of acquisitions $N$ or the sample size $W=\abs{\Omega}$ increases, which hampers interpretability across stacks of different sizes. To address this, we introduce a lower reference level based on the expected objective value under $\BSigma=\Bmrm{I}_N$, i.e.,~a noise floor level that depends on $(N,W)$.

Given a PL method and sample coherence matrix $\BTau$, we denote by
\begin{align*}
    &f\big(\BTau, {\hat{\btheta}}\mkern1mu\big)\quad&&\text{the achieved objective function value};\\
    &f_{\max}\big(\BTau\big)\quad&&\text{an upper bound for }f\big(\BTau, {\hat{\btheta}}\mkern1mu\big)\text{ that is attained in the idealized case where $\BTau=\BSigma$};\\
    &f_{\min}\big(\BTau\big)\quad&&\text{a lower bound for }f\big(\BTau, {\hat{\btheta}}\mkern1mu\big).
\end{align*}
\Cref{tab:gof_bounds} summarizes the expressions and bounds used here; analytical derivations are given in \Cref{prop:gof_bounds}. For ML-PT, the stated upper bound is attained in the idealized case $\BTau=\BSigma$ only when the ML weights are entrywise nonnegative (e.g.,~for AR(1)-type coherence magnitudes; cf.\ \Cref{prop:ar1_inverse}). 
In the presence of negative ML weights, the bound remains valid but not necessarily sharp for $\BTau=\BSigma$. We obtain the coefficient via the two-step normalization
\begin{align}
&\mathcal{F}\big(\BTau, {\hat{\btheta}}\mkern1mu\big) \coloneqq \frac{f\big(\BTau, {\hat{\btheta}}\mkern1mu\big) - f_{\min}\big(\BTau\big)}{f_{\max}\big(\BTau\big) - f_{\min}\big(\BTau\big)}\in\sbrack*{0,1}, \label{eq:intermediatestep}\\
&\mathcal{F}_{\mrm{noise}} \coloneqq \mathbb{E}\sbrack*{\mathcal{F}\big(\BTau, {\hat{\btheta}}\mkern1mu\big)\;\big\vert \;\BSigma=\Bmrm{I}_N},\\
&\gamma_{\mrm{GOF}}\big(\BTau, {\hat{\btheta}}\mkern1mu\big) \coloneqq \max\cbrack*{\frac{\mathcal{F}\big(\BTau, {\hat{\btheta}}\mkern1mu\big) - \mathcal{F}_{\mrm{noise}}}{1 - \mathcal{F}_{\mrm{noise}}}, \;0} \in\sbrack*{0,1}.\label{eq:gofcoef}
\end{align}

For most methods, a fixed lower bound is sufficient, since $\mathcal{F}$ remains, in expectation, substantially smaller under $\BSigma=\Bmrm{I}_N$ than under typical signal scenarios. For ML-ED, however, this ordering can be reversed, so a fixed lower bound may be overly optimistic; in that case, the lower bound must be chosen adaptively as a function of $\BTau$.

The noise floors $\mathcal{F}_{\mrm{noise}}$ are difficult to deduce analytically. Instead, we approximate them via Monte Carlo simulation. Specifically, we evaluate \Cref{eq:intermediatestep} on datasets consisting of circularly symmetric complex Gaussian white noise, for stack sizes $N$ ranging from 4 to 100 acquisitions and sample sizes $W=\abs{\Omega}\label{W}$ ranging from 9 to 7000 SHPs, yielding a total of 400 combinations of stack and sample sizes. For ML methods, sample sizes smaller than $3N$ are omitted, since these methods require inversion of $\BTau$. We then fit the model
\begin{equation}\label{eq:noisemodel}
     \widehat{\mathcal F}_{\mrm{noise}}(N,W)\coloneqq\sigma\big(\alpha+\beta_s N^{-p}+\beta_w W^{-q}+\beta_{sw}N^{-p}W^{-q}\big),
 \end{equation}
 where $\sigma(x)\coloneqq(1+\exp(-x))^{-1}$ is the logistic sigmoid, to the simulated values by nonlinear least squares, obtaining the parameters reported in \Cref{tab:rational_fit_params}. These parameters can be used to approximate $\mathcal{F}_{\mrm{noise}}$ for arbitrary $(N,W)$ within the simulated range. Together with \Cref{tab:gof_bounds} and \Cref{eq:intermediatestep,eq:gofcoef}, this yields a practical procedure for computing the goodness-of-fit coefficient. By construction, the resulting coefficient satisfies \Cref{eq:bestcase,eq:worstcase} via the normalization with respect to the noise floor.

\begin{table}[ht]
\renewcommand{\arraystretch}{1.5}  
\centering
\begin{tabularx}{\textwidth}{l X l l l}
\toprule
\textbf{Method} & \textbf{Weights} &
\boldmath$f\big(\BTau, {\hat{\btheta}}\mkern1mu\big)$ & \boldmath$f_{\max}(\BTau)$ & \boldmath$f_{\min}(\BTau)$ \\
\midrule
PT & all &
$\displaystyle\sum_{i=1}^N \sum_{j=i+1}^N \mrm{W}_{ij} \cos(\phi_{ij} - (\hat{\theta}_i - \hat{\theta}_j))$ & 
$\displaystyle\sum_{i=1}^N \sum_{j=i+1}^N |\mrm{W}_{ij}|$ & 
$0$ \\[6pt]

ED & $\Bmrm{W}^{\mrm{CW}}$ &
$\lambda_{\max}\bigl(\BTau\bigr)$ & $\lambda_{\max}\bigl(\abs{\BTau}\bigr)$ & $1$ \\

ED & $\Bmrm{W}^{\mrm{EW}}$ &
$\lambda_{\max}\bigl(\BPhi\bigr)$ &
$\displaystyle N $ & $1$ \\

ED & $\Bmrm{W}^{\mrm{ML}}$ &
$\lambda_{\max}\bigl(-\lvert\BTau\rvert^{-1}\circ\BTau\bigr)$ &
$-1$ &
$\displaystyle N^{-1}\operatorname{tr}\bigl(-\lvert\BTau\rvert^{-1}\circ\BTau\bigr)$ \\
\bottomrule
\end{tabularx}
\caption{Summary of expressions to compute $\mathcal{F}$ as in \Cref{eq:intermediatestep} for PT and ED
methods with the basic weights discussed in \Cref{sec:basis}. See \Cref{prop:gof_bounds} for a detailed derivation.}
\label{tab:gof_bounds}
\end{table}

\begin{table}[ht]
\centering
\renewcommand{\arraystretch}{1.2}
\setlength{\tabcolsep}{12pt} 
\begin{tabularx}{\textwidth}{X X X X X X X @{}}
\toprule
 & \multicolumn{3}{c}{\textbf{Phase triangulation}}
 & \multicolumn{3}{c}{\textbf{Eigendecomposition}} \\
\cmidrule(lr){2-4} \cmidrule(lr){5-7}
 & $\Bmrm{W}^{\mrm{EW}}$ & $\Bmrm{W}^{\mrm{CW}}$ & $\Bmrm{W}^{\mrm{ML}}$
 & $\Bmrm{W}^{\mrm{EW}}$ & $\Bmrm{W}^{\mrm{CW}}$ & $\Bmrm{W}^{\mrm{ML}}$ \\
\midrule
$\alpha$        & -2.5789 & -2.5108 & -2.7937 & -2.6572 & -22.931 & -1.4451 \\
$\beta_s$       &  5.8898 &  6.1917 &  6.9665 &  5.9320 &  15.603 &  5.9659 \\
$p$             &  0.37061 &  0.37605 &  0.38582 &  0.32869 &  0.056864 &  0.44218 \\
$\beta_w$       &  3.6568 &  3.8003 &  2.9138 &  4.3057 &  21.476 &  8.3068 \\
$q$             &  0.50406 &  0.42735 &  0.18221 &  0.49705 &  0.091900 &  0.45487 \\
$\beta_{sw}$    & -4.7641 & -4.7640 & -2.3920 & -5.7196 & -13.270 & -10.9316 \\
$\operatorname{RMSE}$ & 5.208e-3 & 5.105e-3 & 6.446e-3 & 4.956e-3 & 9.320e-4 & 4.879e-3 \\
\bottomrule
\end{tabularx}
\caption{Fitted parameters and fit accuracy for the empirical noise floor models of PT and ED methods with the basic weights discussed in \Cref{sec:basis}. For each method, the estimated noise level is modeled as in \Cref{eq:noisemodel}. The first six rows list the fitted parameters $\alpha,\beta_s,p,\beta_w,q,\beta_{sw}$ for each method. The final row reports the ordinary root mean squared error ($\operatorname{RMSE}$) of the fitted model against the Monte Carlo estimates of $\mathcal{F}_{\mrm{noise}}$ obtained from stacks of circularly symmetric complex Gaussian noise.}
\label{tab:rational_fit_params}
\end{table}

\subsection{Closure Phase Coefficient}

As described in \Cref{sec:basis}, when $\BTau=\BSigma$ (i.e.,~in the absence of coherence matrix estimation errors), the closure phases vanish. Conversely, nonzero closure phases in the sample coherence matrix indicate either the presence of such errors or violations of the assumption in \Cref{eq:deflinkage} (cf.~\cite{zwieback2016statistical}). This motivates a quality coefficient based on closure phase. We define
\begin{equation}
    {\BPhi}^\Delta\coloneqq{\binom{N}{3}}^{-1}\sum_{i< j< k}\cos\paren*{\phi^\Delta_{ijk}}.
\end{equation}
By the same reasoning as above, ${\BPhi}^\Delta=1$ when $\BTau=\BSigma\neq\Bmrm{I}_N$. In the $\BSigma=\Bmrm{I}_N$ scenario, the phases $\phi_{ij}$ are i.i.d.\ uniform on $(-\pi,\pi]$ and since $\cos\paren{\phi^\Delta_{ijk}}=\cos\paren{\phi_{ij} + \phi_{jk} + \phi_{ki}}$, the following holds:
\begin{equation}\label{eq:ezero}
    \forall\,i <j:\;\phi_{ij}\sim\mathcal{U}(-\pi,\pi] \qquad\xRightarrow{\text{\Cref{irwinhall}}}\qquad \mathbb{E}\sbrack*{{\BPhi}^\Delta}=0.
\end{equation}
Thus, ${\BPhi}^\Delta$ satisfies \Cref{eq:bestcase,eq:worstcase} and can, postponing the verification of \Cref{eq:corrnae}, be used as a quality coefficient. For a triplet of index pairs $(i,j)$, $(j,k)$ and $(i,k)$ that are subject to strong decorrelation, i.e.,~$\abs{\Sigma_{ij}},\abs{\Sigma_{jk}},\abs{\Sigma_{ik}}\ll 1$, nonzero closure phases are expected. Nevertheless, PL methods that incorporate coherence information may still perform accurately. For example, if $\abs{\BTau}$ follows an $\operatorname{AR}(1)$ correlation model, then \Cref{prop:ar1_inverse} implies that the ML weights vanish for all non-adjacent acquisitions, i.e.\ $\mrm{W}_{ij}=0$ for $\abs{i-j}>1$. Thus, the ML-PT objective emphasizes temporally adjacent pairs exclusively and, if the first-superdiagonal entries of $\abs{\BTau}$ are large, accurate estimates may still be obtained. We account for this by introducing a CW variant,
\begin{equation}\label{eq:defwcpcoef}
    {\BTau}^\Delta\coloneqq\frac{\sum_{i< j< k}\abs{\Tau_{ij}\Tau_{jk}\Tau_{ki}}\cos\paren*{\phi^\Delta_{ijk}}}{\sum_{i< j< k}\abs{\Tau_{ij}\Tau_{jk}\Tau_{ki}}}=\frac{\sum_{i< j< k}\Re\paren{\Tau_{ij}\Tau_{jk}\Tau_{ki}}}{\sum_{i< j< k}\abs{\Tau_{ij}\Tau_{jk}\Tau_{ki}}}.
\end{equation}
Here, \Cref{eq:bestcase} still holds, and it can be shown that \Cref{eq:worstcase} holds approximately except for very small sample sizes. We define the equal-weighted and the coherence-weighted closure phase coefficient as
\begin{equation}\label{eq:clippedcpcoef}
\gamma_{\mrm{CP}}\big(\BPhi\big) \coloneqq \max\cbrack[\Big]{{\BPhi}^\Delta, 0} \quad\text{and}\quad \gamma_{\mrm{CPw}}\big(\BTau\big)\coloneqq\max\cbrack[\Big]{{\BTau}^\Delta, 0},
\end{equation}
respectively. These coefficients provide an a priori indicator of PL quality that can be computed directly from the sample coherence matrix. This enables efficient pre-selection of reliable pixels before running the potentially computationally expensive PL.

\subsection{Ambiguity Coefficient}\label{standardinner}

As described in Cao et al.~\cite[p.~2]{cao2015mathematical}, the fundamental idea of ML methods is to maximize the log-likelihood function:
\begin{equation}\label{eq:loglikelihood}
    \hat{\Theta}_{\C} = \argmax_{\bvtheta\in\mathbb{T}^N}\ell(\bvtheta) \quad\text{with}\quad \ell(\bvtheta) \coloneqq \operatorname{log}\paren*{p\paren*{\BTau\;\big\vert \;\BSigma=\abs{\BTau}\circ\bvtheta\bvtheta^\herm}}.
\end{equation}
The uncertainty of ML-PL estimates can be quantified via the relationship between the Hessian of $\ell(\bvtheta)$ at the maximizer (i.e.,~the local curvature of the log-likelihood) and the covariance matrix of the estimate~\cite{zwieback2022reliable}. Intuitively, the sharper the peak of $\ell$ around its maximum, the more sensitive $\ell$ is at its maximum to changes in $\bvtheta$ and thus the smaller the estimation uncertainty.

A natural way to generalize this idea to non-ML methods is to replace $\ell$ by the objective in \Cref{eq:defPLmethods}, which we denote by $f$. Unfortunately, closed-form expressions are generally unavailable, and the Hessian of $f$ is often impractical to compute numerically. We therefore adopt a related rationale: for a high-certainty PL estimate $\hat{\bvtheta}\in\mathbb{T}^N$, the objective $f$ should attain a substantially larger value at $\hat{\bvtheta}$ than at vectors in $\mathbb{T}^N$ that are significantly different from $\hat{\bvtheta}$. To formalize this notion, we consider two vectors $\bvtheta_1,\bvtheta_2\in\mathbb{T}^N$ significantly different if their standard complex inner product vanishes, $\inner*{\bvtheta_1}{\bvtheta_2}=0$, or, equivalently, if $\bvtheta_2$ lies in the orthogonal complement of $\bvtheta_1$, which we denote by $\cbrack{\bvtheta_1}^\perp\label{orthcomplement}$.

In the worst case, $f$ attains the same global maximum at two orthogonal vectors, implying that the set of maximizers $\hat{\Theta}_{\C}$ is not one-dimensional, i.e.,~$\operatorname{dim}(\operatorname{span}(\hat{\Theta}_{\C}))>1$. In the idealized case $\BTau=\BSigma$, this may indicate a nonregular linkage in the sense of \Cref{eq:conditionregular}; for PT, it is in fact an implication. As an example, consider the following sample coherence matrix:
\begin{equation}\label{eq:examplemat}
    \BTau=\begin{pmatrix}
1 & 1 & 1 & 0.12 & 0.20\,\alpha & 0.10\,\alpha^{2} \\
1 & 1 & 1 & 0.10\,\alpha^{2} & 0.09 & 0.05\,\alpha \\
1 & 1 & 1 & 0.05\,\alpha & 0.10\,\alpha^{2} & 0.09 \\
0.12 & 0.10\,\alpha & 0.05\,\alpha^{2} & 1 & 1 & 1 \\
0.20\,\alpha^{2} & 0.09 & 0.10\,\alpha & 1 & 1 & 1 \\
0.10\,\alpha & 0.05\,\alpha^{2} & 0.09 & 1 & 1 & 1
\end{pmatrix}
,
\end{equation}
where $\alpha\coloneqq\exp(2\pi\iu/3)$. Computations show that, for CW-PT, the objective admits a set of global maximizers with more than one dimension: for every $\varphi\in(-\pi,\pi]$, the vector $\hat{\btheta}\paren{\varphi}\coloneqq(0,0,0,\varphi,\varphi,\varphi)^\tran$ maximizes the objective. In particular, the associated vectors in $\mathbb{T}^N$ for $\varphi=0$ and $\varphi=\pi$ are orthogonal in the standard complex inner product. Hence, the PL estimate is completely ambiguous in this example.

For PT methods, an additional difficulty is that finding a global maximizer of the nonconvex problem in \Cref{eq:defPTreal} is NP-hard~\cite[Proposition~3.5]{zhang2006complex}. For this reason, Singer advocates relaxations such as ED~\cite{singer2011angular}. Accordingly, PT is in practice typically solved with iterative methods that guarantee convergence only to a local optimum, and it is implicitly assumed that, for a suitable initialization, the obtained solution is globally optimal~\cite{montiguarnieri2008exploitation,ferretti2011new}. Comparing the resulting estimate to other candidates in the solution space can help diagnose when this assumption fails.

We propose a quality coefficient based on the rationale of the previous paragraphs. Let $\mathcal{X}^{(1)}$ denote the feasible set, i.e.,~$\mathcal{X}^{(1)}=\mathbb{T}^N$ for PT methods and $\mathcal{X}^{(1)}=\mathbb{U}^N$ for ED methods. Define
\begin{equation}\label{eq:deforthEDPT}
     \hat{\Theta}_{\R}^{(1)} \coloneqq \cbrack*{\angle\bvtheta\,\mid\,\bvtheta\in\argmax_{\boldsymbol{v}\in\mathcal{X}^{(1)}}\inner*{\boldsymbol{v}\boldsymbol{v}^\herm}{\Bmrm{W}\circ\BPhi}_\frob},
\end{equation}
i.e., as in \Cref{eq:defPLmethods}, and fix a representative $\hat{\btheta}^{(1)}\in \hat{\Theta}_{\R}^{(1)}$ via \Cref{eq:uniquehistory} for a chosen reference index $p$. Then set
\begin{equation}
\mathcal{X}^{(2)}\coloneqq \mathcal{X}^{(1)}\cap \cbrack*{\exp\big(\iu\hat{\btheta}^{(1)}\big)}^\perp,
\qquad \hat{\Theta}_{\R}^{(2)} \coloneqq \cbrack*{\angle\bvtheta \,\mid\, \bvtheta\in\argmax_{\boldsymbol{v}\in\mathcal{X}^{(2)}} \inner*{\boldsymbol{v}\boldsymbol{v}^\herm}{\Bmrm{W}\circ\BPhi}_\frob},
\end{equation}
and also fix a representative $\hat{\btheta}^{(2)}\in\hat{\Theta}_{\R}^{(2)}$. We refer to $\hat{\btheta}^{(1)}$ and $\hat{\btheta}^{(2)}$ as the primary and secondary PL estimates, respectively. By construction, the secondary estimate is orthogonal to the primary in the complex inner product. Moreover, with $\mathcal{F}$ as defined in \Cref{eq:intermediatestep}, we denote $\mathcal{F}^{(i)}\coloneqq \mathcal{F}\big(\BTau, \hat{\btheta}^{(i)}\big)$.

For high noise levels, $\mathcal{F}^{(1)}$ may be close to the noise floor, and $\mathcal{F}^{(2)}$ may even fall below it. In such cases, a naive coefficient of the form $\paren{\mathcal{F}^{(1)}-\mathcal{F}^{(2)}}\slash\paren{\mathcal{F}^{(1)}-\mathcal{F}_{\mrm{noise}}}$, i.e.,~the relative difference between the two goodness-of-fit values, can become large, although this does not indicate a reliable phase estimate. The dampening constant $\mu>0$ and the use of $\max$ operators in the following definition mitigate this issue:
\begin{equation}\label{eq:ambcoef}
    \gamma_{\mrm{A}}\coloneqq
    \frac{\mathcal{F}^{(1)}-\max\cbrack*{\mathcal{F}^{(2)},\,\mathcal{F}_{\mrm{noise}}}}
         {\max\cbrack*{\mu,\,\mathcal{F}^{(1)}-\mathcal{F}_{\mrm{noise}}}}.
\end{equation}

We call $\gamma_{\mrm{A}}$ the ambiguity coefficient. By construction, it satisfies \Cref{eq:worstcase}. In general, \Cref{eq:bestcase} cannot be guaranteed; however, our experiments indicate that $\gamma_{\mrm{A}}$ captures the rationale above and can serve as a useful criterion for DS-InSAR pixel selection. The parameter $\mu$ is chosen such that $1-\gamma_{\mrm{A}}$ aligns with the normalized absolute error at high noise levels in our experiments; empirically, we find that $\mu=1/3$ yields reasonable results.

In practice, the ambiguity coefficient is easiest to compute for ED methods. The primary PL solution can be obtained using standard numerical routines (e.g.,~LAPACK's \texttt{zheevr}) that return the largest eigenvalue and a corresponding eigenvector. Moreover, the same eigensolver pass can typically provide the second-largest eigenvector at negligible additional cost. In contrast, for PT methods, obtaining an orthogonal secondary solution requires solving a constrained nonlinear optimization problem (see \Cref{sec:solvers}), which is substantially more expensive from a computational perspective.

\section{Experiments}\label{sec:experiments}

\subsection{Data}

To assess \Cref{eq:corrnae} for the proposed quality coefficients, we use simulated data under two decorrelation scenarios: exponential decorrelation and seasonal decorrelation; see~\cites[p.~20]{even2018insar}[p.~4]{mirzaee2023non}. Specifically, we draw a phase history $\btheta\in{(-\pi,\pi]}^N$ with independent components $\theta_i\sim\mathcal{U}(-\pi,\pi]$ and define the corresponding phase vector $\bvtheta\coloneqq\exp(\iu\btheta)$. We then construct a covariance matrix of the form
\begin{equation*}
\Cov=\abs{\Cov}\circ\bvtheta\bvtheta^\herm,
\end{equation*}
where the magnitude $\abs{\Cov}$ encodes temporal and seasonal decorrelation via
\begin{equation}\label{eq:absc_longterm}
\forall\,i,j\in\{1,\dots,N\}:\quad
\abs{\Cov}_{ij} = (\eta_0-\eta_\infty)\,\zeta^{\abs{i-j}}\,\Xi_{ij}+\eta_\infty.
\end{equation}
The term $\Xi_{ij}$ introduces seasonal modulation through
\begin{equation}
\Xi_{ij}\coloneqq\paren*{A + B\cos\paren*{\tfrac{2\pi}{P}i}}\paren*{A + B\cos\paren*{\tfrac{2\pi}{P}j}},
\end{equation}
with
\begin{equation}
A\coloneqq\frac{\sqrt{\xi_{\mrm{high}}} + \sqrt{\xi_{\mrm{low}}}}{2},
\qquad
B\coloneqq\frac{\sqrt{\xi_{\mrm{high}}} - \sqrt{\xi_{\mrm{low}}}}{2}.
\end{equation}

The chosen parameter settings are summarized in \Cref{tab:sim_scenarios}. For each scenario and each noise level $\nu=k/50$ with $k=0,1,\dots,50$, we draw $10^6$ samples from $\mathcal{CN}\big(0,\,(1-\nu)\,\Cov+\nu\,\Bmrm{I}_N,\,0\big)$ and arrange them into an SLC stack of size $(1000\times1000\times N)$, where we use the first SLC as the reference.

\begin{table}[ht]
    \centering
    \renewcommand{\arraystretch}{1.2}
    \begin{tabular}{@{} l c c c c c c c @{}}
        \toprule
        \textbf{Scenario} & N & $\eta_0$ & $\eta_\infty$ & $\zeta$ & $P$ & $\xi_{\mrm{high}}$ & $\xi_{\mrm{low}}$ \\
        \midrule
        
        Exponential decorrelation
        & 50 & 1 & 0.025 & 0.85 & 1 & 1 & 1 \\
        
        Seasonal decorrelation
        & 50 & 1 & 0.05 & 0.9 & 12 & 1 & 0.1 \\
        
        \bottomrule
    \end{tabular}
    \caption{Parameter settings for the simulated decorrelation scenarios. $P=1$ implies no seasonality in the simple case.}
    \label{tab:sim_scenarios}
\end{table}

For the real-data experiments, we use a time series of 44 coregistered TerraSAR-X SLC images over Visp, Switzerland, acquired during the summers of 2017--2022 in descending orbit and HH polarization; see \Cref{tab:acq_dates}. The acquisition of 11 Sep 2019 is used as the reference for subsequent processing. The selected area covers approximately $2\,\mrm{km}\times 2\,\mrm{km}$ and includes an urban part with stable scatterers (e.g.,~roads and buildings) as well as agricultural areas and forests. Apart from man-made structures, the interferometric stack exhibits strong temporal decorrelation, which is evident in the single-look differential interferograms in \Cref{fig:visp_rmli_mapds_infts}. Several agricultural fields and grassland areas, in particular a well-maintained stadium field, are plausible candidates for distributed scatterers. To account for topography, we use SwissALTI3D\footnote{\url{https://www.swisstopo.admin.ch/en/height-model-swissalti3d}} (\textcopyright~swisstopo), a publicly available $0.5\,\mrm{m}$ resolution DEM. 

\begin{figure}
    \centering
    \includegraphics[width=\linewidth]{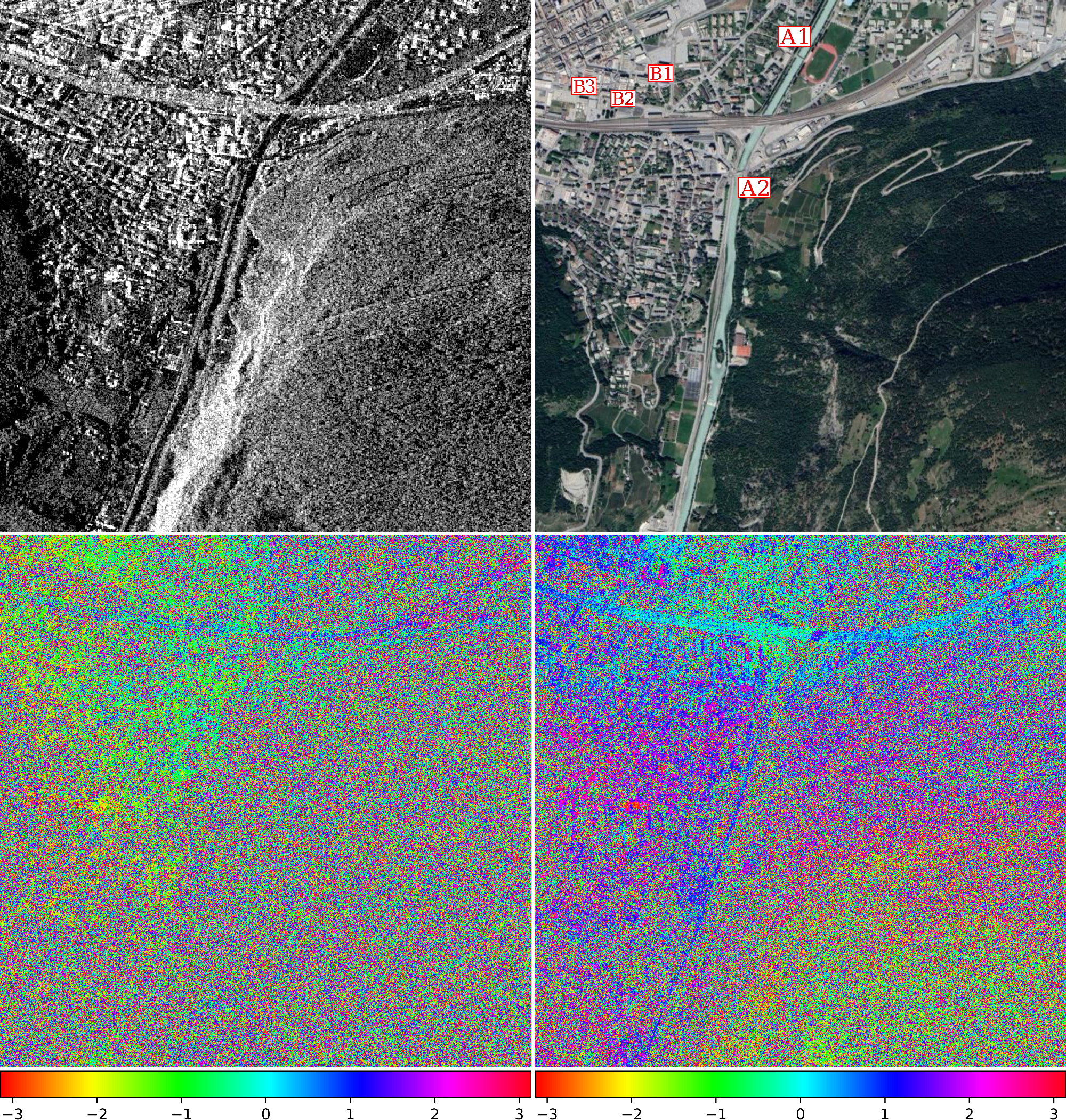}
    \caption{Region of interest within the TerraSAR-X dataset over Visp, Switzerland. Top left: magnitude of the SLC acquired on 11 Sep 2019 (reference). Top right: orthophoto of the scene (\textcopyright~CNES, Spot Image, swisstopo, NPOC, \url{https://map.geo.admin.ch/}) with markers for locations discussed in \Cref{sec:real_results}. Bottom left: single-look differential interferometric phase for 30 Jun 2017 vs.\ 11 Sep 2019 (reference). Bottom right: single-look differential interferometric phase for 31 Aug 2019 vs.\ 11 Sep 2019 (reference).}
    \label{fig:visp_rmli_mapds_infts}
\end{figure}

\subsection{Processing Steps}

The simulated data requires no additional preprocessing. The preprocessing of the real SAR data and the DEM is performed using the GAMMA Remote Sensing software. First, a reference scene is geocoded using simulated SAR images, following~\cite{frey2013dem,wegmullerautomated}. Next, the SwissALTI3D orthometric DEM heights (above the geoid) are converted to ellipsoidal heights using the geoid undulation provided by the Swiss Federal Office of Topography (\textcopyright~swisstopo).\footnote{\url{https://www.swisstopo.admin.ch/en/geoid-en}} The DEM is then mapped into SAR geometry using the lookup table obtained in the geocoding step. The SLCs are subsequently cropped and coregistered to the reference SLC. Next, we subtract from the phases of all SLCs the reference SLC phase and the corresponding simulated topographic interferometric phase, computed from the DEM and precise orbit state vectors. This reduces high-frequency fringes that would otherwise impede coherence estimation. Forming these quasi-interferograms in advance is equivalent to first computing the sample coherence matrix from the SLCs and then simulating and subtracting the topographic phase for each of its entries. The former is, however, considerably faster, since it requires only $N - 1$ simulated phases.

For the resulting real and simulated SLC stacks, processing is performed pixel-wise, i.e.,~each pixel is processed independently of all others.\footnote{In parallel-computing terminology, the workflow is therefore \emph{embarrassingly parallel}.} On an architecture providing 48 parallel threads, we partition the computation into jobs consisting of all pixels at a fixed azimuth position. For each pixel, we follow a workflow similar to Ferretti et al.~\cite{ferretti2011new}. First, we identify all pixels within a $(w\times w)$ neighborhood that are statistically homogeneous and connected to the center pixel~\cite{yao2024phase}, using $w=101$ for the real data and $w=31$ for the simulated data. Unlike Ferretti et al., we employ the Kuiper test~\cite[p.~739]{press2007numerical} to assess statistical homogeneity, with a significance level of $0.25$ ($0.05$ for the simulated data). We classify the center pixel as a DS pixel if the number of neighboring statistically homogeneous pixels (SHPs) is at least $50$. If this criterion is met, we estimate the coherence matrix for the resulting ensemble using \Cref{eq:coherencedef}. We then perform PL on the estimated coherence matrix for each method discussed in \Cref{sec:phaselinking}, and compute the corresponding quality coefficients from \Cref{sec:coefficients}.

For the ambiguity coefficient in \Cref{eq:ambcoef}, we set the dampening parameter to $\mu=0.333$. For ML methods, singular sample coherence matrices must be avoided, since these methods require inversion of $\BTau$. To reduce the likelihood of singularity, we apply ML-PL only when the sample size (i.e.,~the number of SHPs) is at least three times the stack size. In addition, we apply spectral regularization as in~\cite[Eq.~(3)]{zwieback2022cheap} with $\beta=10^{-4}$. When PL is performed, the output magnitudes are the acquisition-wise averages of the SHP magnitudes, while the phases are given by the PL solution, shifted such that the reference acquisition has zero phase, as in \Cref{eq:uniquehistory}. When PL is not performed, we retain the original data and set all quality coefficients to zero.

\subsection{Results}

\subsubsection{Simulated data}\label{sec:sim_results}

\Cref{fig:all_stats_temporal,fig:all_stats_seasonal} illustrate the dependence of the average normalized absolute error $\overline{\mrm{nAE}}$ and the average quality coefficients on the noise-mixing parameter~$\nu$. In addition, \Cref{tab:global_corr_error_both} reports global performance measures pooled over all valid pixels and all simulated noise levels, namely the correlation $\rho=\operatorname{corr}(1-\gamma,\mrm{nAE})$ and the average absolute calibration error $\bar e=\operatorname{avg}\big(\abs{1-\gamma-\mrm{nAE}}\big)$. These statistics are reported to evaluate how well each coefficient tracks the true error for the corresponding estimator (they are not meant as a benchmark comparison of PL methods).

\begin{table}[ht]
\centering
\begin{tabularx}{\textwidth}{@{} l *{12}{>{\small\centering\arraybackslash}X} @{}}
\toprule
Method
& \multicolumn{3}{c}{(E) Correlation $\rho$}
& \multicolumn{3}{c}{(E) Avg.\ abs.\ error $\bar e$}
& \multicolumn{3}{c}{(S) Correlation $\rho$}
& \multicolumn{3}{c}{(S) Avg.\ abs.\ error $\bar e$} \\
\cmidrule(lr){2-4}\cmidrule(lr){5-7}\cmidrule(lr){8-10}\cmidrule(lr){11-13}
& $\rho_{\mrm{GOF}}$ & $\rho_{\mrm{A}}$ & $\rho_{\mrm{CP(w)}}$
& $\bar e_{\mrm{GOF}}$ & $\bar e_{\mrm{A}}$ & $\bar e_{\mrm{CP(w)}}$
& $\rho_{\mrm{GOF}}$ & $\rho_{\mrm{A}}$ & $\rho_{\mrm{CP(w)}}$
& $\bar e_{\mrm{GOF}}$ & $\bar e_{\mrm{A}}$ & $\bar e_{\mrm{CP(w)}}$ \\
\midrule
EW-ED & 0.679 & 0.736 & 0.581 & 0.166 & 0.160 & 0.300
      & 0.798 & 0.894 & 0.674 & 0.140 & 0.123 & 0.262 \\
EW-PT & 0.672 & 0.748 & 0.570 & 0.173 & 0.174 & 0.302
      & 0.793 & 0.895 & 0.672 & 0.142 & 0.124 & 0.262 \\
\midrule
CW-ED & 0.811 & 0.610 & 0.770 & 0.102 & 0.370 & 0.122
      & 0.864 & 0.846 & 0.788 & 0.090 & 0.145 & 0.122 \\
CW-PT & 0.803 & 0.677 & 0.758 & 0.104 & 0.368 & 0.124
      & 0.848 & 0.874 & 0.785 & 0.098 & 0.119 & 0.123 \\
\midrule
ML-ED & 0.866 & -0.104 & 0.770 & 0.098 & 0.725 & 0.119
      & 0.905 & 0.097 & 0.801 & 0.078 & 0.613 & 0.120 \\
ML-PT & 0.834 & 0.103 & 0.769 & 0.156 & 0.678 & 0.119
      & 0.884 & 0.255 & 0.801 & 0.193 & 0.525 & 0.120 \\
\bottomrule
\end{tabularx}
\caption{Simulation results: global correlations and average absolute errors pooled over all valid pixels and all simulated noise levels for the exponential~(E) and seasonal~(S) decorrelation scenarios (\Cref{tab:sim_scenarios}). Here,
$\rho_{\bullet}=\operatorname{corr}(1-\gamma_{\bullet},\,\mrm{nAE})$ and
$\bar e_{\bullet}=\operatorname{avg}(\abs{1-\gamma_{\bullet}-\mrm{nAE}})$.
For EW methods, $\gamma_{\mrm{CP(w)}}=\gamma_{\mrm{CP}}$; otherwise $\gamma_{\mrm{CP(w)}}=\gamma_{\mrm{CPw}}$.}
\label{tab:global_corr_error_both}
\end{table}

\begin{figure}

    \centering
    \includegraphics[width=\linewidth]{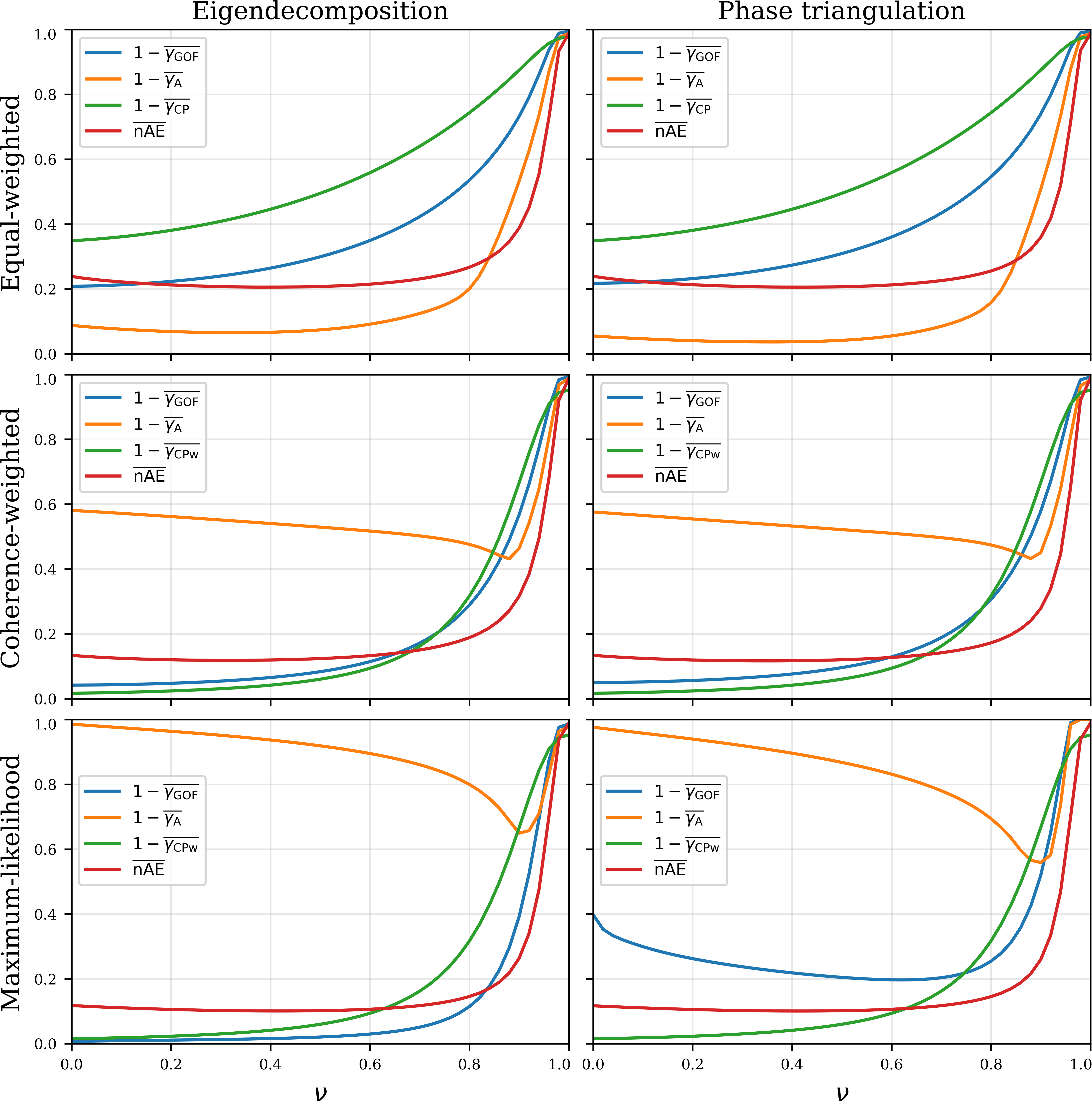}
    \caption{Exponential decorrelation simulation: average normalized absolute error and complementary average quality coefficients as a function of the noise-mixing parameter $\nu$. Each panel corresponds to one PL method class (rows: EW, CW, ML; columns: ED and PT). Curves show the averages $\overline{\mrm{nAE}}$, $1-\overline{\gamma_{\mrm{GOF}}}$, $1-\overline{\gamma_{\mrm{A}}}$, and $1-\overline{\gamma_{\mrm{CP}}}$ (EW) or $1-\overline{\gamma_{\mrm{CPw}}}$ (CW and ML), where the overline denotes the sample average over all simulated pixels.}
    \label{fig:all_stats_temporal}
    
\end{figure}

\begin{figure}

    \centering
    \includegraphics[width=\linewidth]{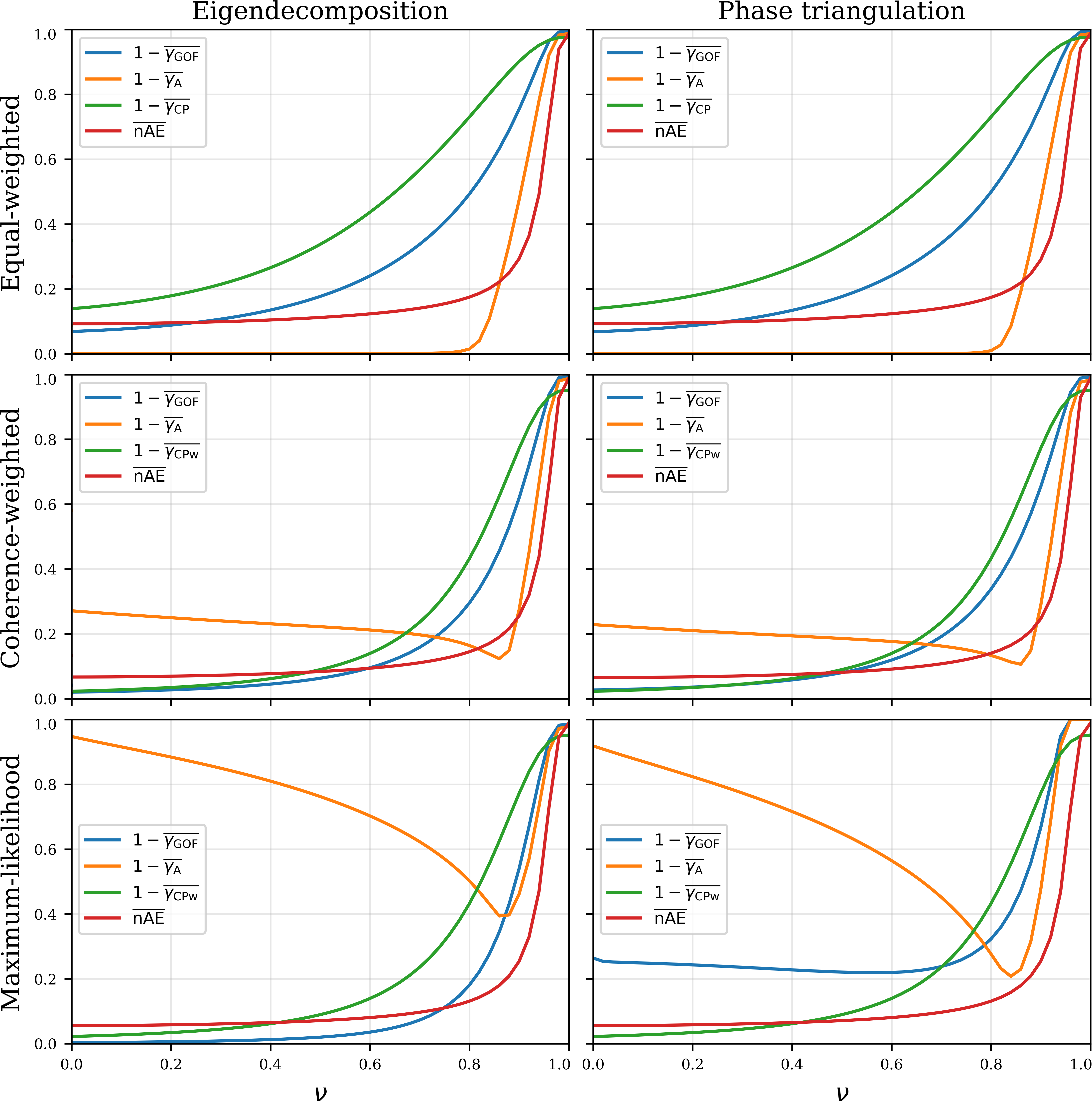}
    \caption{Seasonal decorrelation simulation: average normalized absolute error and complementary average quality coefficients as a function of the noise-mixing parameter $\nu$. Each panel corresponds to one PL method class (rows: EW, CW, ML; columns: ED and PT). Curves show the averages $\overline{\mrm{nAE}}$, $1-\overline{\gamma_{\mrm{GOF}}}$, $1-\overline{\gamma_{\mrm{A}}}$, and $1-\overline{\gamma_{\mrm{CP}}}$ (EW) or $1-\overline{\gamma_{\mrm{CPw}}}$ (CW and ML), where the overline denotes the sample average over all simulated pixels.}
    \label{fig:all_stats_seasonal}
    
\end{figure}

Across both scenarios, the goodness-of-fit coefficient shows the strongest and most consistent association with $\mrm{nAE}$. In the exponential decorrelation scenario, $\rho_{\mrm{GOF}}$ is about $0.67$-$0.68$ for EW methods, increases to about $0.80$-$0.81$ for CW methods, and reaches $0.87$ for ML-ED (\Cref{tab:global_corr_error_both}). The corresponding calibration errors $\bar e_{\mrm{GOF}}$ follow the same trend, with the smallest value attained by ML-ED ($\bar e_{\mrm{GOF}}=0.098$). In the seasonal decorrelation scenario, the same pattern holds, with $\rho_{\mrm{GOF}}$ up to $0.91$ and $\bar e_{\mrm{GOF}}$ down to $0.078$ . These trends are also reflected in \Cref{fig:all_stats_temporal,fig:all_stats_seasonal}, where $1-\overline{\gamma_{\mrm{GOF}}}$ closely tracks $\overline{\mrm{nAE}}$ across~$\nu$.

The closure phase coefficient provides a second, comparatively stable predictor. For CW and ML configurations, $\rho_{\mrm{CPw}}$ remains high (about $0.76$-$0.80$ across both scenarios) and the corresponding calibration errors are small ($\bar e_{\mrm{CPw}}\approx 0.12$; \Cref{tab:global_corr_error_both}). In the EW case, by contrast, $\rho_{\mrm{CP}}$ is noticeably lower (about $0.57$-$0.67$) and $\bar e_{\mrm{CP}}$ is substantially larger (about $0.26$-$0.30$).

The ambiguity coefficient behaves differently. For EW and CW methods, it correlates positively with $\mrm{nAE}$ in both scenarios, with particularly strong correlations in the seasonal case (\Cref{tab:global_corr_error_both}). For ML methods, the behavior is mixed: $\rho_{\mrm{A}}$ is negative for ML-ED in the exponential scenario and only weakly positive otherwise. This is consistent with the fact that $\gamma_{\mrm{A}}$ is constructed to satisfy the worst-case behavior of \Cref{eq:worstcase} but does not guarantee the best-case behavior of \Cref{eq:bestcase}. In \Cref{fig:all_stats_temporal,fig:all_stats_seasonal}, this remains visible through the characteristic breakpoint in $1-\overline{\gamma_{\mrm{A}}}$ induced by the dampening parameter~$\mu$ in \Cref{eq:ambcoef}.

Finally, the per-$\nu$ averages confirm that all methods approach a common noise-dominated regime as $\nu\to 1$, with $\overline{\mrm{nAE}}\approx 1$ and average quality coefficients close to zero (cf.~\Cref{eq:worstcase}). At $\nu=0$, the seasonal scenario yields smaller average errors than the exponential scenario (e.g.,~for EW-ED, $\overline{\mrm{nAE}}$ decreases from $0.238$ to $0.092$), with correspondingly larger average quality coefficients, consistent with the higher base coherence simulated in the seasonal scenario (\Cref{tab:sim_scenarios}).

\subsubsection{TerraSAR-X stack}\label{sec:real_results}

We first summarize basic coverage statistics for DS processing. With an SHP threshold of 50 and a significance level of 0.25 for the Kuiper test, about $57.5\%$ of the pixels are classified as DS pixels and enter coherence matrix estimation and PL. The phase-linked differential interferograms for the first acquisition relative to the reference acquisition are shown in \Cref{fig:ALL_INTF}. In the urban part of the scene, the single-look interferogram (\Cref{fig:visp_rmli_mapds_infts}) already exhibits a coherent phase pattern, and all methods yield very consistent phase-linked estimates. In strongly decorrelated regions, most notably the large agricultural field near the scene center (location A2 in \Cref{fig:visp_rmli_mapds_infts}) and surrounding vegetated areas, differences between the PL results become pronounced. By contrast, the well-maintained stadium field (location A1) appears comparatively consistent across methods, showing only minor visible differences in the reconstructed phase.

\begin{figure}

    \centering
    \includegraphics[width=0.925\linewidth]{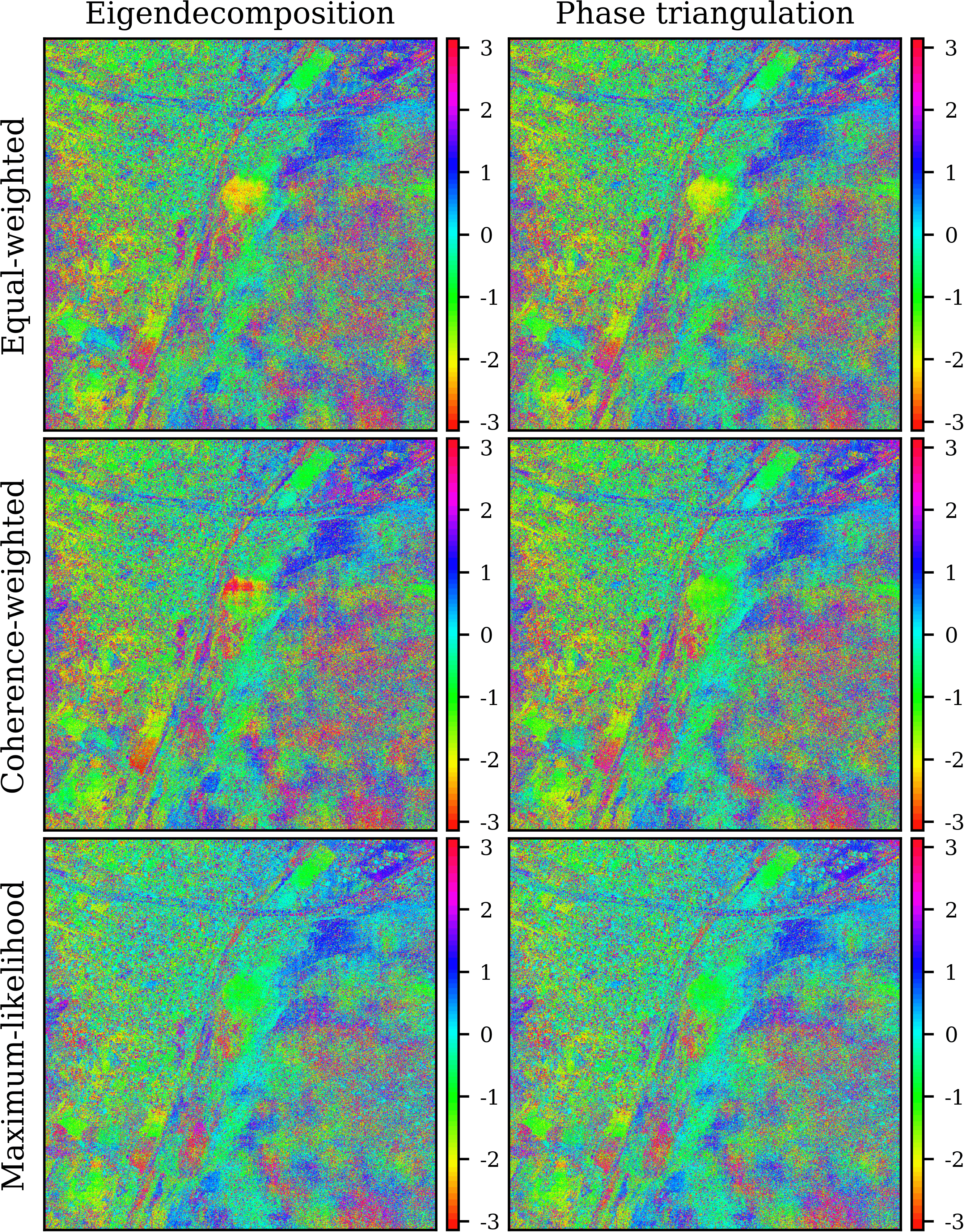}
\caption{Phase-linked differential interferometric phase for 30 Jun 2017 relative to the reference acquisition on 11 Sep 2019. The bottom-left panel of \Cref{fig:visp_rmli_mapds_infts} shows the corresponding unfiltered single-look interferogram. Panels are organized by weighting scheme (rows: EW, CW, ML) and solver class (columns: ED, PT). For each method, the interferometric phase of each pixel is replaced by the PL estimate whenever the pixel is classified as DS according to the number of SHPs; otherwise, the original interferometric phase is retained.}
    \label{fig:ALL_INTF}
    
\end{figure}

\Cref{fig:phase_closure_with_hist} shows the closure phase coefficients. The EW coefficient $\gamma_{\mrm{CP}}$ yields generally low values over large parts of the scene (median $\approx 0.11$ on DS pixels), whereas the CW version $\gamma_{\mrm{CPw}}$ produces higher values in areas with strong coherence (median $\approx 0.32$). Spatially, both coefficients assign high values to stable surfaces and low values to strongly decorrelated regions. In particular, regions such as the river surface and densely vegetated areas exhibit very low closure phase coefficients, while the urban area and some agricultural parcels attain substantially higher values.

\begin{figure}
    \centering
    \includegraphics[width=0.9\linewidth]{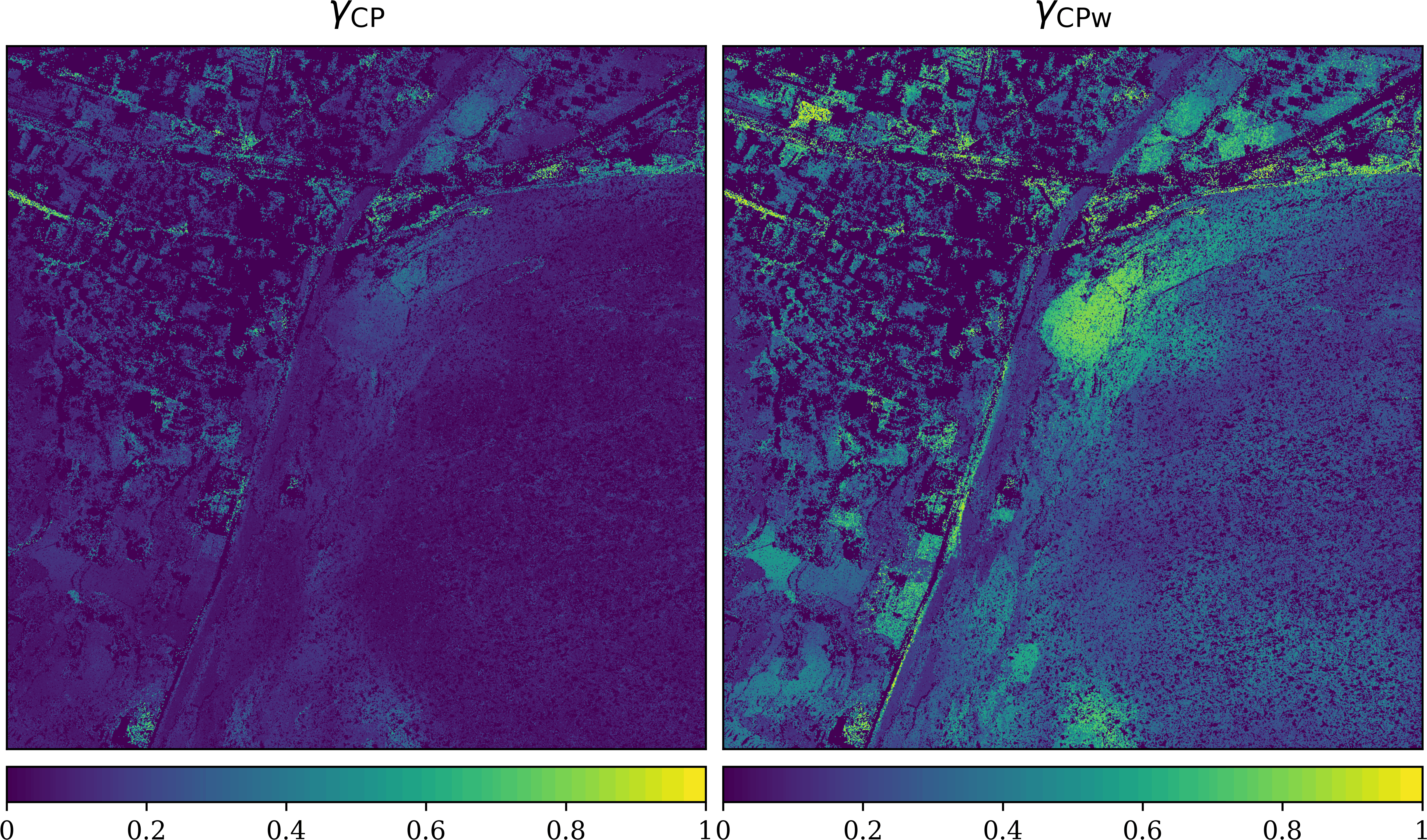}
    \caption{Spatial distribution of the EW (left) and CW (right) closure phase coefficients for the TerraSAR-X stack over Visp, Switzerland. Values are shown in $[0,1]$ with colorbar, where higher values indicate smaller (weighted) closure phases and thus better consistency with the assumption of a phase history.}
    \label{fig:phase_closure_with_hist}
\end{figure}

The spatial distribution of the goodness-of-fit coefficient is shown in \Cref{fig:all_gof}. High $\gamma_{\mrm{GOF}}$ values concentrate in the urban area and on other stable surfaces, whereas vegetated areas exhibit substantially lower values. CW methods yield higher goodness-of-fit values than EW methods over most of the scene. Among the ML methods, ML-ED produces the highest overall goodness-of-fit values on pixels where it is applicable (median $\gamma_{\mrm{GOF}}\approx0.51$ on DS pixels), while ML-PT yields substantially lower values (median $\gamma_{\mrm{GOF}}\approx0.25$).

\Cref{fig:all_amb} shows the ambiguity coefficient for the Visp scene. It reveals a clear contrast between two prominent DS regions: the stadium field (A1) and the agricultural parcels (A2) marked in \Cref{fig:visp_rmli_mapds_infts}. A1 is regularly maintained, whereas A2 undergoes substantial intra- and inter-annual changes (crop growth, harvesting, and varying parcel management). Correspondingly, A2 exhibits low $\gamma_{\mrm{A}}$ values and coincides with the region where the PL estimates in \Cref{fig:ALL_INTF} differ most across configurations, while A1 shows comparatively high $\gamma_{\mrm{A}}$ values despite only moderate goodness-of-fit values in some cases. \Cref{fig:comparison_a} provides a close-up: for CW-PT, $\gamma_{\mrm{CPw}}$ and $\gamma_{\mrm{GOF}}$ are similar over A1 and A2 (with $\gamma_{\mrm{CPw}}$ slightly higher in A2), whereas $\gamma_{\mrm{A}}$ is markedly higher over A1, indicating complementary reliability information beyond fit and closure phase consistency. Similar behavior is observed across the other PL configurations. A second example is shown in \Cref{fig:comparison_b}, which relates the coefficients to visible urban change at three parking lots (B1--B3). The orthophotos show that B1 is already present throughout the SAR acquisition period, while B2 is constructed between summer 2018 and summer 2019 and B3 between summer 2020 and summer 2021. At B2 and B3, the pre-construction orthophotos show grass, so the pre-construction acquisitions are expected to have lower coherence than the post-construction acquisitions.  For CW-PT, $\gamma_{\mathrm{CPw}}$ and $\gamma_{\mathrm{GOF}}$ again attain broadly similar magnitudes across the three locations (with $\gamma_{\mathrm{CPw}}$ highest at B3), whereas $\gamma_{\mathrm{A}}$ is much higher at B1 and B2 than at B3. This is consistent with interpreting $\gamma_{\mathrm{A}}$ as an indicator of ambiguity: after the late construction at B3, the high-coherence acquisitions can dominate the CW-PT objective, leaving the earlier (pre-construction) phases comparatively weakly constrained and allowing competing explanations. Qualitatively similar behavior is observed for the other PL configurations; an exception is CW-ED, for which $\gamma_{\mathrm{A}}$ may remain high when $\BTau$ is effectively dominated by a single coherent block. Lastly, unlike the other coefficients, $\gamma_{\mathrm{A}}$ is consistently low over the layover area in the bottom center, where we expect unreliable estimates.\pagebreak

\begin{figure}

    \centering
    \includegraphics[width=0.975\linewidth]{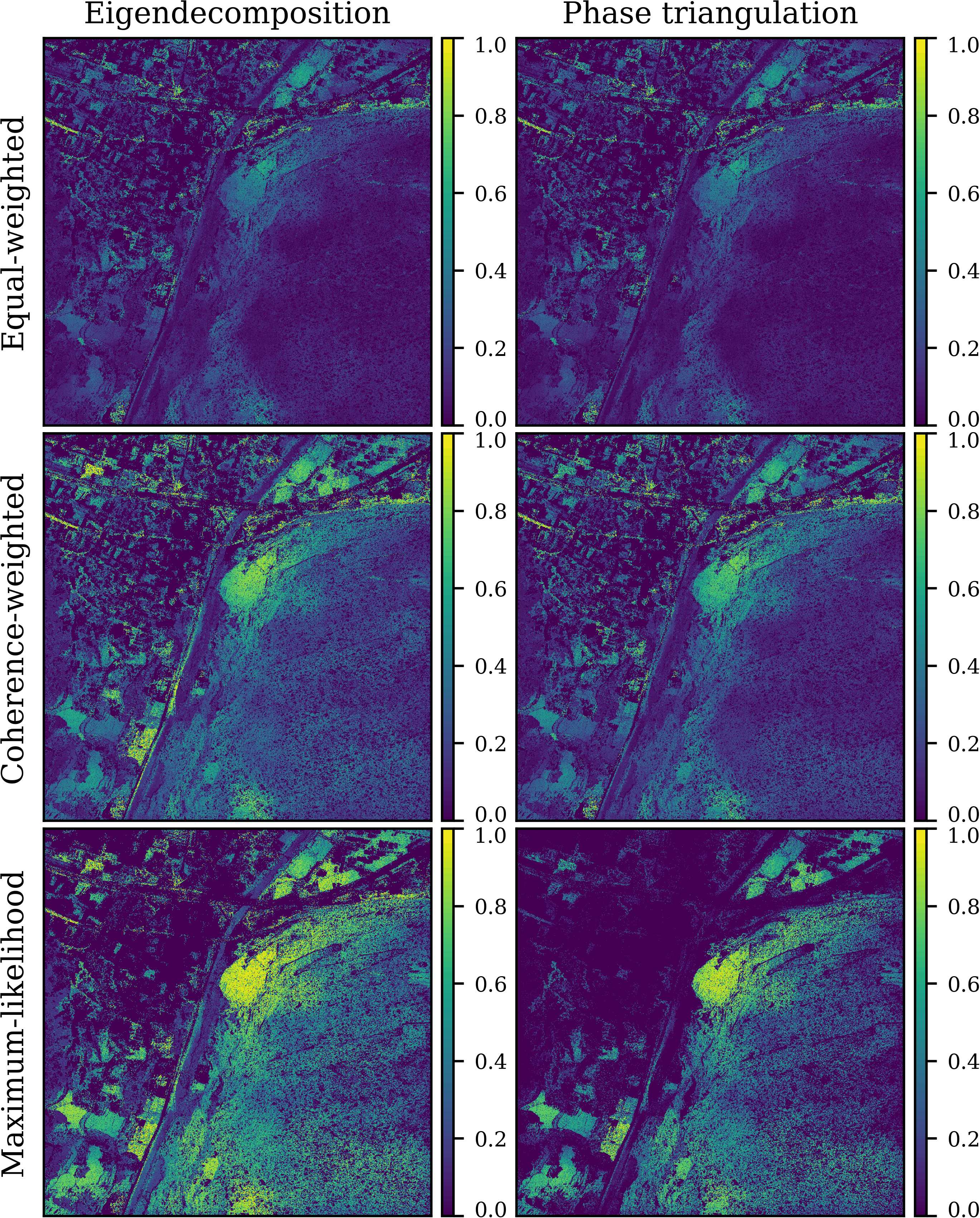}
\caption{Spatial distribution of the goodness-of-fit coefficient $\gamma_{\mrm{GOF}}$ for the TerraSAR-X stack over Visp, Switzerland, shown for six PL configurations (rows: EW, CW, ML; columns: ED and PT). Each panel displays the per-pixel values of $\gamma_{\mrm{GOF}}\in[0,1]$, where higher values indicate better agreement between the estimated rank-one phase model and the observed sample coherence. Color bars are shown for each panel.}
    \label{fig:all_gof}
    
\end{figure}

\begin{figure}

    \centering
    \includegraphics[width=0.975\linewidth]{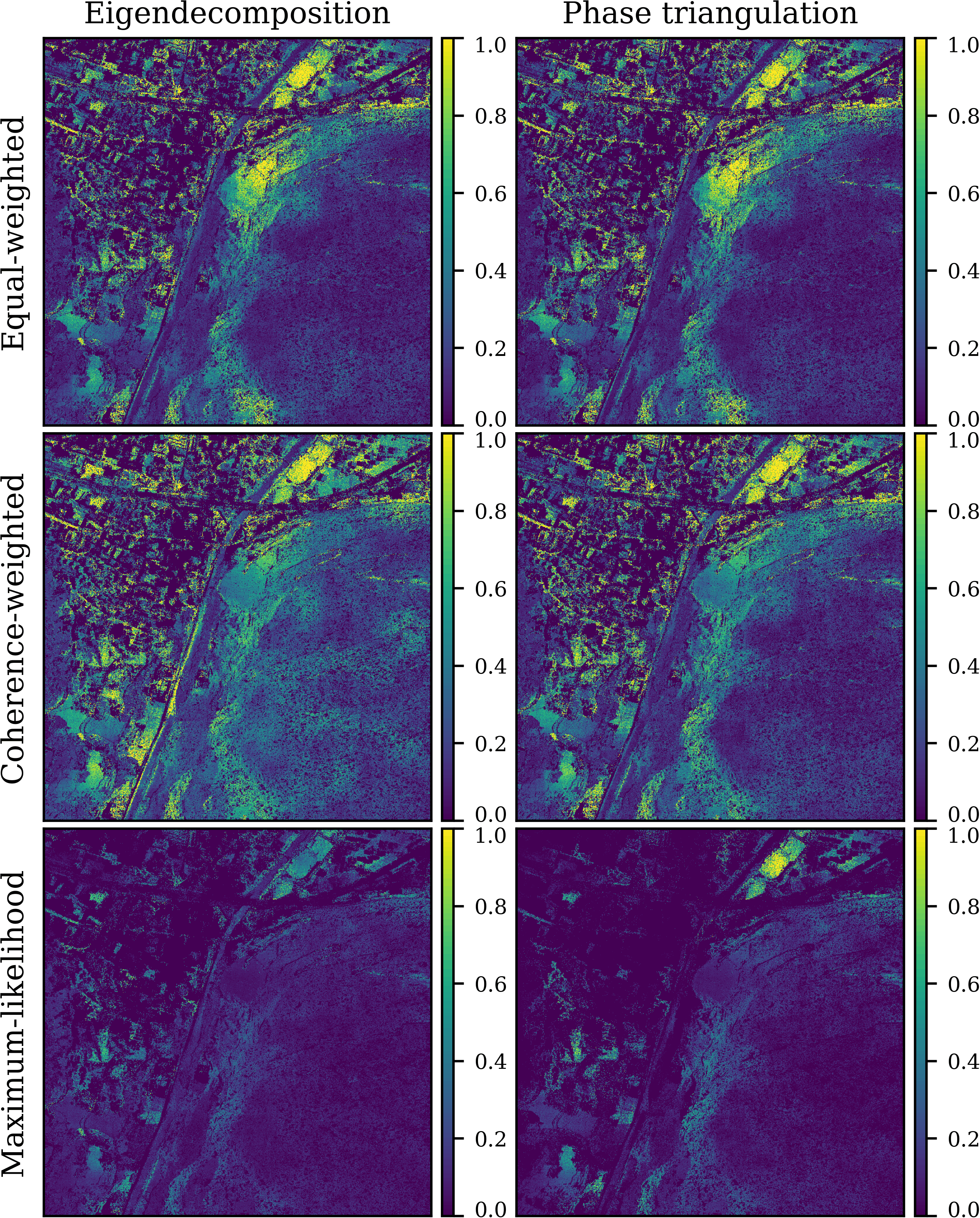}
    \caption{Spatial distribution of the ambiguity coefficient $\gamma_{\mrm{A}}$ for the TerraSAR-X stack over Visp, Switzerland, shown for six PL configurations (rows: EW, CW, ML; columns: ED and PT). Each panel displays the per-pixel values of $\gamma_{\mrm{A}}\in[0,1]$, where higher values indicate a larger separation in terms of the achieved objective value between the primary solution and the best alternative in the orthogonal complement. Color bars are shown for each panel.}
    \label{fig:all_amb}
    
\end{figure}

\begin{figure}
    \centering
    \includegraphics[width=\linewidth]{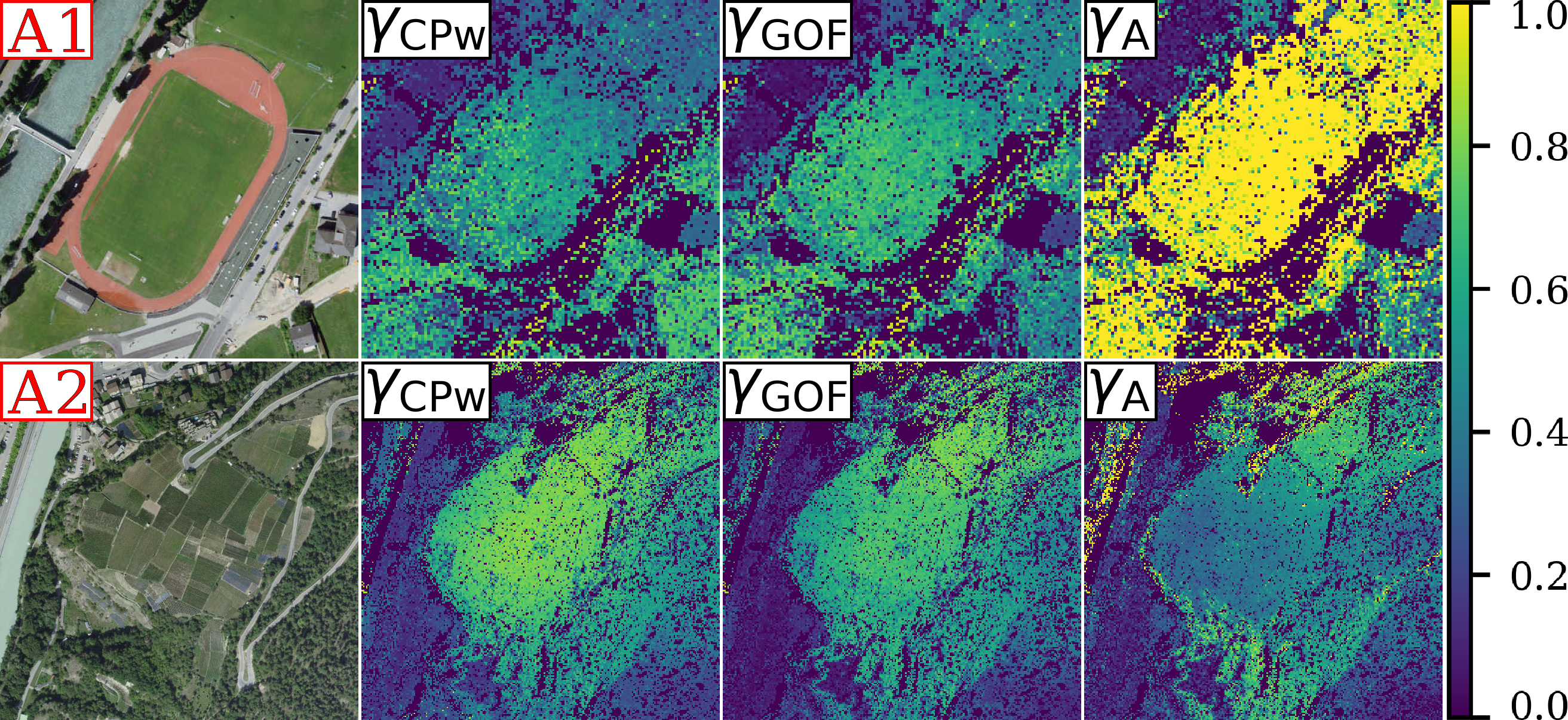}
    \caption{Comparison of CW-PT quality coefficients on two distributed scatterer regions in the TerraSAR-X Visp stack (A1: stadium field; A2: agricultural parcels; locations marked in  \Cref{fig:visp_rmli_mapds_infts}). Columns show (left to right): Swissimage orthophoto (\textcopyright~swisstopo: \url{https://www.swisstopo.admin.ch/en/orthoimage-swissimage-10}), coherence-weighted closure phase coefficient $\gamma_{\mathrm{CPw}}$, goodness-of-fit coefficient $\gamma_{\mathrm{GOF}}$, and ambiguity coefficient $\gamma_{\mathrm{A}}$ for coherence-weighted phase triangulation (CW-PT). Rows correspond to A1 (top) and A2 (bottom). All coefficient panels share a common colorbar and are shown on the scale $[0,1]$.}
    \label{fig:comparison_a}
\end{figure}

\begin{figure}
    \centering
    \includegraphics[width=\linewidth]{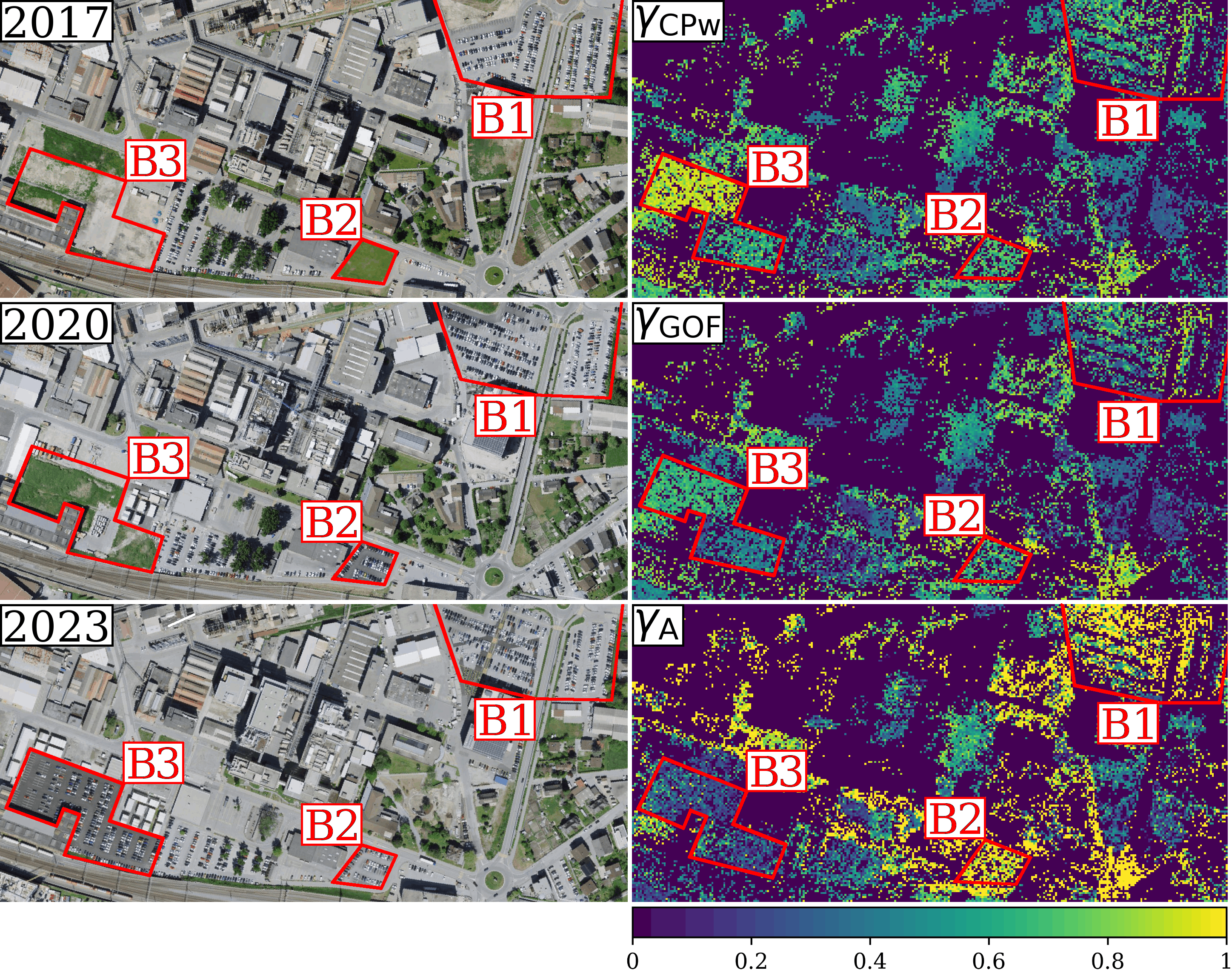}
    \caption{Urban close-up illustrating how the ambiguity coefficient reflects scene changes in the TerraSAR-X Visp dataset (locations B1--B3 marked; see  \Cref{fig:visp_rmli_mapds_infts}). The left column shows Swissimage orthophotos for three years (top to bottom: 2017, 2020, 2023; \textcopyright~swisstopo: \url{https://www.swisstopo.admin.ch/en/orthoimage-swissimage-10}) as visual context for land-cover change. The orthophotos do not coincide with the SAR acquisition dates; the TerraSAR-X scenes were acquired in the early morning hours (around 05:30 local time). The right column shows quality-coefficient maps for coherence-weighted phase triangulation (CW-PT) computed from the TerraSAR-X stack (top to bottom: coherence-weighted closure phase coefficient $\gamma_{\mathrm{CPw}}$, goodness-of-fit coefficient $\gamma_{\mathrm{GOF}}$, and ambiguity coefficient $\gamma_{\mathrm{A}}$). All coefficient panels share a common colorbar and are shown on the scale $[0,1]$. The marked DS locations correspond to parking lots with different construction times: B1 is already present at the beginning of the acquisition period (summers 2017--2022), B2 is constructed between summer 2018 and summer 2019, and B3 between summer 2020 and summer 2021 (grass before construction in both cases). Accordingly, coherence is expected to be higher in post-construction epochs than in pre-construction epochs for B2 and B3. While $\gamma_{\mathrm{CPw}}$ and $\gamma_{\mathrm{GOF}}$ are of comparable magnitude at these locations (with $\gamma_{\mathrm{CPw}}$ highest at B3), $\gamma_{\mathrm{A}}$ is markedly lower at B3 than at B1/B2. This is consistent with late structural change: the strongly coherent post-construction acquisitions are emphasized by CW-PT, which can leave pre-construction phases comparatively weakly constrained and thus increase solution ambiguity. Similar patterns are observed for the other PL configurations; CW-ED can be an exception when $\BTau$ is effectively dominated by a single coherent block, which may still yield a dominant eigenvalue.}
\label{fig:comparison_b}
\end{figure}

\section{Discussion}

The main contributions of this work are three complementary scalar quality coefficients for assessing the reliability of PL estimates, formulated in a unified framework for a range of PL methods. The experiments indicate that these coefficients capture distinct, practically relevant failure modes: lack of model fit (goodness-of-fit), internal inconsistency of the observed interferometric phases (closure phase), and weak distinction of the recovered solution (ambiguity). The experiments are designed to evaluate whether the proposed coefficients behave as intended under controlled decorrelation models and to illustrate their spatial behavior on one real DInSAR stack.

A first observation from the simulations is that $\gamma_{\mrm{GOF}}$ is the most consistently informative coefficient with respect to the normalized absolute error $\mrm{nAE}$ across both decorrelation scenarios and all method classes. The strong correlations and comparatively small calibration errors suggest that $\gamma_{\mrm{GOF}}$ is a useful default indicator when a single coefficient is desired. At the same time, the real-data maps show that high $\gamma_{\mrm{GOF}}$ values concentrate where the PL model is plausible (urban structures and other stable areas), while strongly decorrelated surfaces (dense vegetation, river) are consistently assigned low values.

The closure phase coefficient conveys information similar to that of $\gamma_{\mrm{GOF}}$, as reflected by a strong correlation. For example, in the real-data experiment we have $\operatorname{corr}(\gamma_{\mrm{GOF}},\gamma_{\mrm{CPw}})=0.897$ for the CW-PT method. However, the close-ups in \Cref{fig:comparison_a,fig:comparison_b} also reveal clear differences. Importantly, the closure phase coefficient is computed directly from $\BTau$ and can therefore be used a priori, before running any solver. The CW version is particularly well aligned with CW and ML objectives: it down-weights triplets that are already expected to be unreliable, and empirically it shows a stronger association with $\mrm{nAE}$ than the EW version outside the EW method class. In practice, this makes $\gamma_{\mrm{CP(w)}}$ attractive as a fast pre-screening criterion, especially when large stacks make PL computationally expensive. In applications where line-of-sight displacement is to be estimated, the closure phase coefficient can be used with a comparatively strict threshold. By contrast, if parameters that may drive nonzero closure phases (e.g.,~soil moisture) are also of interest and should be estimated jointly, then filtering based on closure phase may be inappropriate and should be omitted.

The ambiguity coefficient $\gamma_{\mrm{A}}$ targets a different aspect: it measures how clearly the primary solution is distinguished, in terms of the achieved objective value, from competing alternatives in the solution space. This is most useful in cases where the goodness-of-fit coefficient alone can be misleading, e.g.,~when multiple competing solutions explain the data almost equally well. In the simulations, $\gamma_{\mrm{A}}$ correlates well with $\mrm{nAE}$ for EW and CW methods in both decorrelation scenarios (\Cref{tab:global_corr_error_both}).
For ML configurations, the relationship is weaker and less stable; for example, $\rho_{\mrm{A}}$ is negative for ML-ED in the exponential case, while it is only weakly positive for ML methods in the seasonal case. This is consistent with the fact that $\gamma_{\mrm{A}}$ is constructed to satisfy \Cref{eq:worstcase} but does not guarantee \Cref{eq:bestcase}. In other words, $\gamma_{\mrm{A}}$ approaches zero as $\BSigma \to \Bmrm{I}_N$ (pure noise), but it is not necessarily one in the absence of coherence matrix estimation errors. Nevertheless, for EW and CW methods, the coefficient still captures useful information. The role of the dampening parameter $\mu$ is visible as a characteristic breakpoint (change in slope) in $1-\overline{\gamma_{\mrm{A}}}$. It marks the transition between the regime where $\mathcal{F}^{(1)}-\mathcal{F}_{\mrm{noise}}$ is small and the denominator is clamped, and the regime where the denominator follows $\mathcal{F}^{(1)}-\mathcal{F}_{\mrm{noise}}$. Increasing $\mu$ shifts this transition towards smaller $\nu$ values (higher SNR), consistent with the intended stabilization at high noise. In the real-data experiments, $\gamma_{\mrm{A}}$ attains low values in areas where the PL outputs differ most visibly across methods despite high goodness-of-fit coefficients. Conversely, it also highlights areas that appear very stable despite only moderate $\gamma_{\mrm{GOF}}$ values. Moreover, the results suggest that $\gamma_{\mrm{A}}$ is sensitive to both land-cover changes and layover, indicating that it captures information beyond pure model fit.

From an operational perspective, the coefficients can be combined in a simple decision logic. $\gamma_{\mrm{CP(w)}}$ can serve as an inexpensive pre-screening criterion, applied with a lenient threshold (i.e.,~a low cutoff) to remove only clearly unreliable pixels before running PL. Among the remaining DS pixels, $\gamma_{\mrm{GOF}}$ can serve as the primary selection criterion, since it most consistently reflects reliability in our simulations, while $\gamma_{\mrm{A}}$ provides an additional safeguard by flagging solutions that are weakly distinguished (and therefore potentially unstable) even when $\gamma_{\mrm{GOF}}$ is high. Note that computing $\gamma_{\mrm{A}}$ for PT is more costly than for ED. In time-critical settings, one may therefore complement the goodness-of-fit coefficient of a PT method with the ambiguity coefficient computed from the corresponding ED method, especially when the ED estimate is used to initialize the PT solver. This is supported by the results that show a high similarity in $\gamma_{\mrm{A}}$ between PT methods and their corresponding ED relaxations. In practice, thresholds for all coefficients are application-dependent and should be calibrated to balance coverage against reliability: stricter thresholds reduce false inclusions but may exclude usable pixels, whereas looser thresholds increase coverage at the risk of retaining unreliable estimates.

Several limitations and open directions remain. First, the presented framework focuses on estimators that can be expressed in the form of \Cref{eq:defPLmethods}. While this covers most common PL approaches, methods such as SBAS are only related indirectly; extending the quality-coefficient framework to explicitly accommodate such estimators would be valuable. Second, the goodness-of-fit coefficient depends on a choice of objective representative for each method. We evaluate each estimator using the objective it is designed to optimize (cf.~\Cref{eq:defPTreal,eq:defEDmethods}), but this choice is not unique: any strictly monotone transformation of the objective yields the same set of maximizers and could therefore also serve as the basis for a quality coefficient. Different representatives can nevertheless lead to different calibrations once the achieved value is normalized between upper and lower reference levels. Identifying objective transformations that yield better calibrated, more monotone, or more method-robust reliability indicators, while preserving the maximizers, appears to be a promising direction for future work. Third, the real-data evaluation is limited to a single TerraSAR-X stack over Visp, Switzerland; validating transferability across sensors and decorrelation regimes (e.g.,~Sentinel-1 C-band, ALOS-2 L-band) and across more diverse land-cover types (ideally incorporating land-cover classification) is an important direction for future work. Finally, the proposed coefficients are reliability indicators conditional on a chosen estimator; using them to support systematic method selection is an interesting direction that we do not pursue here.

\section{Conclusion}

We presented a unified mathematical framework for interferometric phase linking and introduced three quality coefficients: a goodness-of-fit coefficient $\gamma_{\mrm{GOF}}$, a closure phase coefficient $\gamma_{\mrm{CP(w)}}$, and an ambiguity coefficient $\gamma_{\mrm{A}}$. Simulations under exponential and seasonal decorrelation show that the goodness-of-fit coefficient tracks the normalized absolute phase error most consistently, while the closure phase coefficient provides a robust complementary indicator that can be applied prior to phase linking. The ambiguity coefficient captures a different failure mode by highlighting regions where competing solutions exist or where solver outcomes are likely unstable; in the simulations it is most informative for EW and CW configurations, while its behavior under ML objectives is weaker and less stable. On a TerraSAR-X stack over Visp, Switzerland, the spatial patterns of the coefficients align with expected scattering behavior: stable man-made structures yield high values, while strongly decorrelated surfaces yield low values. Importantly, the ambiguity coefficient attains low values in areas where PL estimates diverge most strongly across methods and captures localized land-cover changes. Overall, the proposed coefficients provide practical tools for DS pixel selection and for diagnosing PL reliability, and they are compatible with a broad range of commonly used PL methods within the unified framework.

\section*{Author Contributions}

Conceptualization, M.H. and O.F.; methodology, M.H. and O.F.; software, M.H.; validation, M.H. and O.F.; formal analysis, M.H.; investigation, M.H.; resources, I.H. and O.F.; data curation, M.H.; writing---original draft preparation, M.H.; writing---review and editing, M.H. and O.F.; visualization, M.H.; supervision, O.F. and I.H.; project administration, O.F.; funding acquisition, O.F. All authors have read and agreed to the published version of the manuscript.

\section*{Funding sources}

This work was funded by Swiss National Science Foundation (SNSF) grant number 216526 (\url{https://data.snf.ch/grants/grant/216526}). 

\section*{Data availability}

The TerraSAR-X data of Visp used in this study were provided
by the German Aerospace Center (DLR) via an announcement of opportunities (MTH3910). \Cref{tab:acq_dates} lists the acquisition dates that were used for the experiments.

\section*{Declaration of competing interest}

The authors declare that they have no known competing finan-
cial interests or personal relationships that could have appeared to
influence the work reported in this paper.

\section*{Acknowledgments}

We would like to thank Stefano Tebaldini, Dirk Lebiedz and Moritz Feuerle for fruitful discussions.

\section*{Abbreviations and Mathematical Notation}
The following abbreviations (in alphabetical order) and mathematical notation (in order of appeareance) are used in this manuscript:

\noindent
\begin{tabbing}
\hspace{2.5cm}\= \kill
AR(1) \> AutoRegressive model of order 1\\
BFGS \> Broyden--Fletcher--Goldfarb--Shanno\\
CAESAR \> Component extrAction and sElection Synthetic Aperture Radar\\
CNES \> Centre National d'Ètudes Spatiales\\
COFI-PL \> COvariance-FItting Phase Linking\\
CW \> Coherence-Weighted\\
DEM \> Digital Elevation Model\\
DLR \> Deutsches Zentrum f\"ur Luft- und Raumfahrt\\
DS \> Distributed Scatterer\\
ED \> Eigendecomposition\\
EMI \> Eigendecomposition-based Maximum-likelihood estimator of Interferometric phase\\
ETH \> Swiss Federal Institute of Technology Zurich (ETH Zurich)\\
EW \> Equal-Weighted\\
InSAR \> Interferometric Synthetic Aperture Radar\\
LAPACK \> Linear Algebra PACKage\\
LaMIE \> Large-dimensional Multipass InSAR phase Estimation for Distributed Scatterers\\
ML \> Maximum-Likelihood\\
MM \> Majorize-Minimization\\
NPOC \> National Point of Contact\\
PRCG \> Polak-Ribière nonlinear Conjugate Gradient\\
PL \> Phase Linking\\
PS \> Persistent Scatterer\\
PT \> Phase Triangulation\\
PTCM \> Phase Triangulation Coherence Maximization\\
SAR \> Synthetic Aperture Radar\\
SBAS \> Small BAseline Subset\\
SHP \> Statistically Homogeneous Pixel\\
SLC \> Single-Look Complex\\
SNSF \> Swiss National Science Foundation\\
\>\\
$N$ \> \Cref{sec:cohmat}, number of acquisitions (SLC images)\\
$\Omega$ \> \Cref{sec:cohmat}, set of pixel-wise $N$-dimensional time series belonging to a particular DS\\
$\boldsymbol{\omega}$ \> \Cref{eq:ccg}, multitemporal SLC vector, element of $\Omega$\\
$\mathcal{CN}$ \> \Cref{eq:ccg}, complex normal distribution\\
$\Cov,\ \Cove_{ij}$ \> \Cref{eq:ccg}, covariance matrix and its $(i,j)$-entry\\
$\BSigma,\ \Sigma_{ij}$ \> \Cref{eq:coherencedef}, coherence matrix and its $(i,j)$-entry\\
$\BTau,\ \Tau_{ij}$ \> \Cref{eq:coherencedef}, sample coherence matrix and its $(i,j)$-entry\\
$\abs{\cdot}$ \> \Cref{sec:cohmat}, magnitude (modulus); for matrices/vectors, applied entrywise\\
$\circ$ \> \Cref{sec:cohmat}, Hadamard (entrywise) product\\
$\mathbb{T}^N$ \> \Cref{sec:cohmat}, $N$-torus of unit-modulus $\C^N$-vectors\\
$\Theta_\R,\ \Theta_\C$ \> \Cref{eq:deflinkage}, (real/complex) linkage of $\BSigma$\\
$\btheta,\ \bvtheta$ \> \Cref{eq:deflinkage}, (real/complex) phase history\\
$(\cdot)^\herm$ \> \Cref{eq:deflinkage}, Hermitian (conjugate) transpose\\
$\angle(\cdot)$ \> \Cref{sec:cohmat}, complex argument (phase) in $(-\pi,\pi]$; for matrices/vectors, applied entrywise\\
$\operatorname{span}(\cdot)$ \> \Cref{eq:conditionregular}, linear span of a set of vectors\\
$\dim(\cdot)$ \> \Cref{eq:conditionregular}, dimension of a vector space\\
$\phi_{ij}$ \> \Cref{eq:phasematdef}, interferometric phase (argument of $\Tau_{ij}$, with $\phi_{ij}=0$ if $\Tau_{ij}=0$)\\
$\BPhi,\ \Phi_{ij}$ \> \Cref{eq:phasematdef}, unit phasor matrix and its $(i,j)$-entry, $\Phi_{ij}=\exp(\iu\,\phi_{ij})$\\
$\abs{\Tau_{ij}}$ \> \Cref{eq:phasematdef}, sample coherence magnitude\\
$\phi^\Delta_{ijk}$ \> \Cref{eq:defclosurephase}, closure phase\\
$\hat{\Theta}_\R$ \> \Cref{eq:defPLmethods}, estimator of the real linkage (set of estimated real phase histories)\\
$\mathcal{X}$ \> \Cref{eq:defPLmethods}, feasible set for the opt. problem\\
$\mathbb{U}^N$ \> \Cref{eq:defPLmethods}, unit-norm vectors with nonzero entries, $\mathbb{U}^N=\cbrack*{\boldsymbol{v}\in\C^{N}\mid \norm{\boldsymbol{v}}_2=1\ \wedge\ \forall\,i:\,v_i\neq 0}$\\
$\norm{\cdot}_2$ \> \Cref{eq:defPLmethods}, Euclidean (vector) $2$-norm\\
$\inner*{\Bmrm{A}}{\Bmrm{B}}_\frob$ \> \Cref{eq:defPLmethods}, Frobenius inner product, $\inner*{\Bmrm{A}}{\Bmrm{B}}_\frob=\sum_{i,j}\overline{\mathrm{A}_{ij}}\,\mathrm{B}_{ij}$\\
$\Bmrm{W},\,\mathrm{W}_{ij}$ \> \Cref{sec:phaselinking}, symmetric weight matrix and its $(i,j)$-entry\\
$\lambda_{\max}(\Bmrm{A})$ \> \Cref{eq:defEDmethods}, largest eigenvalue of $\Bmrm{A}$\\
$E_{\lambda_{\max}}(\Bmrm{A})$ \> \Cref{eq:defEDmethods}, eigenspace associated with $\lambda_{\max}(\Bmrm{A})$\\
$\mathds{1}_N$ \> \Cref{sec:sectionweights}, all-ones matrix (identity element for the Hadamard product)\\
$\Bmrm{W}^{\mathrm{ML}}$ \> \Cref{sec:sectionweights}, ML weights, $\Bmrm{W}^{\mathrm{ML}}=-\abs{\BTau}^{-1}\circ\abs{\BTau}$\\
$\Bmrm{W}^{\mathrm{CW}}$ \> \Cref{sec:sectionweights}, coherence weights, $\Bmrm{W}^{\mathrm{CW}}=\abs{\BTau}$\\
$\Bmrm{W}^{\mathrm{EW}}$ \> \Cref{sec:sectionweights}, equal weights, $\Bmrm{W}^{\mathrm{EW}}=\mathds{1}_N$\\
$\Bmrm{M},\,\mathrm{M}_{ij}$ \> \Cref{sec:sectionweights}, mask matrix and its $(i,j)$-entry\\
$\hat{\btheta}$ \> \Cref{sec:coefficients}, estimated real phase history\\
$\operatorname{AE}(\hat{\btheta})$ \> \Cref{eq:defae}, absolute (angular) error, $\operatorname{AE}(\hat{\btheta})=\|\varepsilon(\hat{\btheta})\|_2$\\
$\varepsilon(\hat{\btheta})$ \> \Cref{eq:defae}, wrapped componentwise phase error, $\varepsilon(\hat{\btheta})=\angle\exp(\iu(\hat{\btheta}-\btheta))$\\
$\mathbb{E}[\cdot]$ \> \Cref{eq:msecovbias}, expectation operator\\
$\operatorname{Cov}(\cdot)$ \> \Cref{eq:msecovbias}, covariance operator\\
$\operatorname{tr}(\cdot)$ \> \Cref{eq:msecovbias}, trace operator\\
$\operatorname{nAE}(\hat{\btheta})$ \> \Cref{eq:defnae}, normalized absolute error, $\operatorname{nAE}(\hat{\btheta})=\sqrt{3/(N\pi^2)}\,\operatorname{AE}(\hat{\btheta})$\\
$\Bmrm{I}_N$ \> \Cref{eq:worstcase}, $N\times N$ identity matrix\\
$\operatorname{corr}(\cdot,\cdot)$ \> \Cref{eq:corrnae}, correlation coefficient\\
$\mathcal{F}(\BTau,\hat{\btheta})$ \> \Cref{eq:intermediatestep}, normalized achieved objective value\\
$\gamma_{\mathrm{GOF}}(\BTau,\hat{\btheta})$ \> \Cref{eq:gofcoef}, goodness-of-fit coefficient\\
$\mathcal{F}_{\mathrm{noise}}$ \> \Cref{eq:gofcoef}, noise floor level, $\mathcal{F}_{\mathrm{noise}}=\mathbb{E}[\mathcal{F}(\BTau,\hat{\btheta})\mid \BSigma=\Bmrm{I}_N]$\\
$W$ \> \Cref{sec:gof}, sample size (number of SHPs), $W=\abs{\Omega}$\\
$\widehat{\mathcal{F}}_{\mathrm{noise}}(N,W)$ \> \Cref{eq:noisemodel}, fitted noise floor model\\
$\operatorname{RMSE}$ \> \Cref{tab:rational_fit_params}, root mean squared error\\
$\mathcal{U}$ \> \Cref{eq:ezero}, continuous uniform distribution\\
$\Re(\cdot)$ \> \Cref{eq:defwcpcoef}, real part\\
$\gamma_{\mathrm{CP}}(\BPhi)$ \> \Cref{eq:clippedcpcoef}, equal-weighted closure phase coefficient\\
$\gamma_{\mathrm{CPw}}(\BTau)$ \> \Cref{eq:clippedcpcoef}, coherence-weighted closure phase coefficient\\
$\hat{\Theta}_\C$ \> \Cref{eq:loglikelihood}, estimator of the complex linkage (set of estimated complex phase histories)\\
$\ell(\bvtheta)$ \> \Cref{eq:loglikelihood}, log-likelihood objective\\
$\operatorname{log}(\cdot)$ \> \Cref{eq:loglikelihood}, natural logarithm\\
$p(\cdot\mid\cdot)$ \> \Cref{eq:loglikelihood}, likelihood / conditional density\\
$\inner*{\boldsymbol{u}}{\boldsymbol{v}}$ \> \Cref{standardinner}, standard inner product on $\C^N$\\
$(\cdot)^\perp$ \> \Cref{standardinner}, orthogonal complement (w.r.t.\ the standard complex inner product)\\
$\gamma_{\mathrm{A}}$ \> \Cref{eq:ambcoef}, ambiguity coefficient\\
$\mu$ \> \Cref{eq:ambcoef}, dampening constant used in $\gamma_{\mathrm{A}}$\\
$\nu$ \> \Cref{sec:experiments}, noise level in the covariance mixture $(1-\nu)\,\Cov+\nu\,\Bmrm{I}_N$\\
$\operatorname{avg}(\cdot)$ \> \Cref{tab:global_corr_error_both}, arithmetic mean over the indicated pixel set\\
\end{tabbing}

\appendix

\section{Phase triangulation solvers}\label{sec:solvers}

While ED can be solved efficiently with standardized algorithms that do not require initialization and typically converge to a global optimum, PT relies on iterative solvers whose convergence guarantees are, at best, local. Consequently, a careful choice of solver and initialization is essential for reliable estimates. When available, a convenient strategy is to initialize PT with the output of the corresponding ED method. A standard solver for unconstrained nonlinear optimization is the Broyden--Fletcher--Goldfarb--Shanno (BFGS) algorithm~\cite[Chapter~6.1]{nocedal2006numerical}, as suggested by Ferretti et al.~\cite{ferretti2011new}. Vu et al.~\cite{vu2025covariance} propose an MM solver and Riemannian optimization as further options. A simple iterative method for the special case of the ML estimator is given in~\cite{montiguarnieri2008exploitation}:
\begin{equation}\label{eq:iterationtebaldini}
\theta_p^{(k+1)}\coloneqq\angle\left\{-\sum_{q\neq p}\mrm{W}_{pq}\Phi_{pq} \exp\!\big(\iu\theta_q^{(k)}\big)\right\}.
\end{equation}
Since the original reference provides only a brief derivation, we provide a short interpretation of this iteration in standard numerical-optimization terms. The update can be viewed as a modified Jacobi iteration: the objective in \Cref{eq:defPTreal} is equivalent to $\boldsymbol{v}^\herm\paren*{\Bmrm{W}\circ\BPhi}\boldsymbol{v}$, as shown in \Cref{lemma:realvalpt}. Setting the gradient of this expression with respect to $\boldsymbol{v}$ to zero yields the linear system $(\Bmrm{W}\circ\BPhi)\,\boldsymbol{v} = 0$, for which the Jacobi iteration (cf.~\cite[Eq.~(4.4)]{saad2007iterative}) is
\begin{equation}
v_p^{(k+1)}=\frac{-1}{\mrm{W}_{pp}}\sum_{q\neq p}\mrm{W}_{pq}\Phi_{pq}v_q^{(k)}.
\end{equation}
Updating only the phases, the factor $\mrm{W}_{pp}$ can be omitted and we obtain \Cref{eq:iterationtebaldini}.

We seek a method that makes minimal assumptions and admits a simple implementation, so that it can be extended to \Cref{eq:deforthEDPT}. The MM algorithm and the Jacobi iteration are straightforward to implement. However, the MM algorithm relies on the Hessian $\Bmrm{W}\circ\BPhi$ being positive semidefinite. Similarly, the Jacobi iteration relies on $\Bmrm{W}\circ\BPhi$ being diagonally dominant. These conditions are no longer guaranteed once we penalize the objective to enforce the additional orthogonality constraint. As an efficient, robust alternative that does not rely on such assumptions, we use the Polak-Ribière nonlinear conjugate gradient (PRCG) method~\cite[p.~121f.]{nocedal2006numerical}. In practice, for similar runtimes, PRCG yields slightly better results than Jacobi in strongly decorrelated areas. However, the most significant advantage of the PRCG method is that we can use it to compute the secondary solution needed to calculate the ambiguity coefficient. In the context of \Cref{eq:deforthEDPT}, we enforce the orthogonality constraint via a penalty approach, solving the following problem:
\begin{align} \label{eq:PTpen}
&\argmax_{\boldsymbol{\theta}\in{(-\pi,\pi]}^N} \left\{f\big(\BTau, \boldsymbol{\theta}\mkern1mu\big) - p\abs*{\inner*{\exp\!\big(\iu\btheta\big)}{\exp\!\big(\iu\hat{\btheta}^{(1)}\big)}}^2\right\}\nonumber\\=&\argmax_{\boldsymbol{\theta}\in{(-\pi,\pi]}^N}\left\{f\big(\BTau, \boldsymbol{\theta}\mkern1mu\big)-p\,\exp\!\big(\iu\btheta\big)^\herm\exp\!\big(\iu\hat{\btheta}^{(1)}\big)\exp\!\big(\iu\hat{\btheta}^{(1)}\big)^\herm\exp\!\big(\iu\btheta\big)\right\}\nonumber\\
=&\argmax_{\boldsymbol{\theta}\in{(-\pi,\pi]}^N}\left\{\sum_{i=1}^N\sum_{j=i+1}^N\mrm{W}_{ij}\cos\big(\phi_{ij} - (\theta_i-\theta_j)\big) - p\sum_{i=1}^N\sum_{j=i+1}^N\cos\big(({\hat{\theta}^{(1)}}_i-{\hat{\theta}^{(1)}}_j) - (\theta_i-\theta_j)\big)\right\}.
\end{align}

We multiply the penalty factor $p$ by two at each iteration, which results in \Cref{alg:secondary_phase_linking}. As initial values for the primary and secondary PL estimates, we use the eigenvectors associated with the largest and second-largest eigenvalues of $\Bmrm{W}\circ\BPhi$, respectively, since they are orthogonal by construction.

\vfill

\begin{algorithm}[ht]
\caption{Secondary phase linking via penalized PRCG }
\label{alg:secondary_phase_linking}
\setstretch{1.5}
\BlankLine

\KwIn{$\BPhi$, $\Bmrm{W}$, $\hat{\btheta}^{(1)}$, $\hat{\btheta}^{(2)}_{\text{init}}$, $\kappa$, $L$}
\KwOut{$\hat{\btheta}^{(2)}$ as in \Cref{eq:deforthEDPT}}

\BlankLine

$p \,\leftarrow\, \frac{2}{N(N - 1)}\sum_{i=1}^N\sum_{j=i+1}^N\mrm{W}_{ij}\cos\big(\phi_{ij} - (\hat{\theta}^{(1)}_i-\hat{\theta}^{(1)}_j)\big)$\;
$\hat{\btheta}^{(2)} \,\leftarrow\, \hat{\btheta}^{(2)}_{\text{init}}$\;
$l \,\leftarrow\, 0$\;

\BlankLine

\While{$l \,\le\, L$}{ 
  $\hat{\btheta}^{(2)} \;\leftarrow\; 
     \textbf{PRCG\_solve}\bigl(\text{Eq.\,\eqref{eq:PTpen}},\,\hat{\btheta}^{(2)}\bigr)$\,;
     \tcp*[f]{solve \eqref{eq:PTpen} by PRCG, initialized at $\hat{\btheta}^{(2)}$}

  \If{$\dfrac{1}{N}\,\abs*{\inner[\big]{\exp\!\big(\iu\hat{\btheta}^{(2)}\big)}{\exp\!\big(\iu\hat{\btheta}^{(1)}\big)}}
       \;<\; \kappa$}{
       \textbf{break} \tcp*[r]{orthogonality reached}
  }
  $p \,\leftarrow\, 2\,p$            \tcp*[r]{increase penalty}
  $l \,\leftarrow\, l + 1$\;
}
\end{algorithm}

\section{Linear Algebra, Stochastics and Optimization}

This section collects background facts used in the main text. Many of these results are well known, and we make no claim of originality. We include them to keep the manuscript self-contained and to present the required statements in the notation used throughout the paper. References are provided where appropriate.

\begin{lemma}\label[lemma]{lemma:cohposdef}
The sample coherence matrix as defined in \Cref{eq:coherencedef} is positive semidefinite.    
\end{lemma}

\begin{proof}
We define \[\boldsymbol{\mrm{D}}\coloneqq\operatorname{diag}\begin{pmatrix}
        \sqrt{\textstyle\sum_{\boldsymbol{\omega}\in\Omega}\abs*{\omega_1}^2} & \cdots & \sqrt{\textstyle\sum_{\boldsymbol{\omega}\in\Omega}\abs*{\omega_N}^2}
    \end{pmatrix}\in\R^{N\times N}.\]
Then $\BTau = \boldsymbol{\mrm{D}}^{-1}\big(\sum_{\boldsymbol{\omega}\in\Omega}\boldsymbol{\omega}\boldsymbol{\omega}^\herm\big)\boldsymbol{\mrm{D}}^{-1}$ and thus, for all $\boldsymbol{z}\in\C^N$,
\[
\boldsymbol{z}^\herm\BTau\boldsymbol{z}
= \sum_{\boldsymbol{\omega}\in\Omega}
\boldsymbol{z}^\herm \boldsymbol{\mrm{D}}^{-1}\boldsymbol{\omega}\boldsymbol{\omega}^\herm \boldsymbol{\mrm{D}}^{-1}\boldsymbol{z}
= \sum_{\boldsymbol{\omega}\in\Omega}\abs{\boldsymbol{\omega}^\herm \boldsymbol{\mrm{D}}^{-1}\boldsymbol{z}}^2
\ge 0,
\]
which proves that $\BTau$ is positive semidefinite.
\end{proof}

\begin{lemma}\label[lemma]{lemma:uniquedecomp}
Let $\Bmrm{A} \in \C^{N\times N}$ be Hermitian with $\operatorname{rank}(\Bmrm{A})=1$. Then:
\begin{itemize}
    \item There exist $\lambda\in\R$ and $\boldsymbol{v}\in\C^N$ with $\norm*{\boldsymbol{v}}_2=1$ such that
    $\Bmrm{A}=\lambda\, \boldsymbol{v}\boldsymbol{v}^\herm$.
    \item If $\operatorname{tr}(\Bmrm{A})>0$, then there exists $\boldsymbol{w}\in\C^N$ such that
    $\Bmrm{A}=\boldsymbol{w}\boldsymbol{w}^\herm$.
    \item If the diagonal entries of $\Bmrm{A}$ are all one, then choosing $\boldsymbol{w}$ as any column of
    $\Bmrm{A}$ yields $\Bmrm{A}=\boldsymbol{w}\boldsymbol{w}^\herm$ and $\boldsymbol{w}\in\mathbb{T}^N$.
\end{itemize}
\end{lemma}

\begin{proof}
Since $\Bmrm{A}$ is Hermitian, it admits an ED
$\Bmrm{A}=\Bmrm{U}\Bmrm{D}\Bmrm{U}^\herm$,
where $\Bmrm{U}\in\C^{N\times N}$ is unitary and $\Bmrm{D}\in\R^{N\times N}$ is diagonal
(see~\cite[Thm.~2.5.6]{horn2013matrix}).
Because $\operatorname{rank}(\Bmrm{A})=1$, $\Bmrm{A}$ has exactly one nonzero eigenvalue
$\lambda\in\R$. Let $\boldsymbol{u}$ be a corresponding unit eigenvector. Then
$\Bmrm{A}=\lambda\,\boldsymbol{u}\boldsymbol{u}^\herm$,
which proves the first claim with $\boldsymbol{v}\coloneqq \boldsymbol{u}$.

Moreover, $\operatorname{tr}(\Bmrm{A})$ equals the sum of the eigenvalues, hence
$\operatorname{tr}(\Bmrm{A})>0$ implies $\lambda>0$.
Setting $\boldsymbol{w}\coloneqq \sqrt{\lambda}\,\boldsymbol{u}$ yields
$\Bmrm{A}=\boldsymbol{w}\boldsymbol{w}^\herm$, proving the second claim.

For the third claim, note that $\operatorname{diag}(\Bmrm{A})=\boldsymbol{1}$ implies
$\operatorname{tr}(\Bmrm{A})=N>0$, so by the second claim there exists $\boldsymbol{w}$ with
$\Bmrm{A}=\boldsymbol{w}\boldsymbol{w}^\herm$.
Then $1=\mrm{A}_{ii}=\abs{w_i}^2$ for all $i$, so $\boldsymbol{w}\in\mathbb{T}^N$.
Let $\boldsymbol{a}^{(p)}$ denote the $p$-th column of $\Bmrm{A}$. Since
$\Bmrm{A}=\boldsymbol{w}\boldsymbol{w}^\herm$, we have
$\boldsymbol{a}^{(p)}=\boldsymbol{w}\,\overline{w_p}$, and therefore
\[
\boldsymbol{a}^{(p)}(\boldsymbol{a}^{(p)})^\herm
= \boldsymbol{w}\,\overline{w_p}\, w_p\,\boldsymbol{w}^\herm
= \abs{w_p}^2\,\boldsymbol{w}\boldsymbol{w}^\herm
= \Bmrm{A},
\]
because $\abs{w_p}=1$. Hence choosing $\boldsymbol{w}=\boldsymbol{a}^{(p)}$ (any column) yields
$\Bmrm{A}=\boldsymbol{w}\boldsymbol{w}^\herm$.
\end{proof}

\begin{proposition}\label[proposition]{proposition:rankone}
The following equivalences and implications hold in the context of \Cref{sec:cohmat}:
\begin{align*}
    &\operatorname{rank}\big(\BTau\big) = 1 \; &&\Leftrightarrow
\quad \exists\,\bvtheta\in\mathbb{T}^N:\;\BTau = \bvtheta \bvtheta^{\herm} \; &&\Leftrightarrow
\quad \forall\,i,j:\;\abs{\Tau_{ij}}=1 \\ &\;\Downarrow && &&
\\ &\operatorname{rank}\big(\BPhi\big) = 1 \; &&\Leftrightarrow
\quad \exists\,\bvtheta\in\mathbb{T}^N:\;\BPhi = \bvtheta \bvtheta^{\herm} \; &&\Leftrightarrow
\quad \forall\,i,j,k:\;\phi^\Delta_{ijk}=0.
\end{align*}
\end{proposition}

\begin{proof}
    Since $\BTau$ and $\BPhi$ are Hermitian with unit diagonal, the first equivalence in both lines follows from \Cref{lemma:uniquedecomp}. The rest is proven by the following deductions: \begin{align}&\mrm{I})&&\exists\,\bvtheta\in\mathbb{T}^N:\BTau = \bvtheta \bvtheta^\herm\quad &&\Rightarrow\quad \forall\,i,j:\abs{\Tau_{ij}}=\abs*{\vartheta_i\overline{\vartheta_j}}=1\nonumber\\
    &\mrm{II})&&\forall\,i,j:\abs{\Tau_{ij}}=1\quad &&\Rightarrow\quad \exists\,\boldsymbol{\alpha}\in\C^N\;\forall\,\boldsymbol{\omega}\in\Omega\;\forall\,i:\omega_{i}=\alpha_{i}\omega_{1}\quad\text{(Cauchy-Schwarz)}\nonumber\\& && &&\Rightarrow\quad \exists\,\boldsymbol{\alpha}\in\C^N\;\forall\,i,j:\Tau_{ij}=\frac{\alpha_i\overline{\alpha_j}}{\abs*{\alpha_i}\abs*{\alpha_j}}\;\,=\vartheta_i\overline{\vartheta_j}\quad\text{for}\quad \vartheta_i\coloneqq\frac{\alpha_i}{\abs*{\alpha_i}}\nonumber\\& && &&\Rightarrow\quad \exists\,\bvtheta\in\mathbb{T}^N:\BTau = \bvtheta\bvtheta^\herm\nonumber\\ 
    &\mrm{III})&&\forall\,i,j:\abs{\Tau_{ij}}=1\quad &&\Rightarrow\quad \BTau=\abs{\BTau}\circ\BPhi=\BPhi\quad\Rightarrow\quad \mrm{rank}\big(\BPhi\big)=\mrm{rank}\big(\BTau\big)=1\nonumber\\ 
    &\mrm{IV})\quad &&\exists\,\bvtheta\in\mathbb{T}^N:\BPhi = \bvtheta \bvtheta^\herm&&\Rightarrow\quad\forall\,i,j,k:\Phi_{ij}\Phi_{jk}\Phi_{ki}=\vartheta_i\,\overline{\vartheta_j}\,\vartheta_j\,\overline{\vartheta_k}\,\vartheta_k\,\overline{\vartheta_i}=1\quad\Rightarrow\quad \forall\,i,j,k:\phi^\Delta_{ijk}=0\nonumber\\
    &\mrm{V})&&\forall\,i,j,k:\ \phi^\Delta_{ijk}=0
&&\Rightarrow\quad \forall\,i,j,k:\ \Phi_{ij}\Phi_{jk}=\Phi_{ik} \nonumber\\
& && &&\Rightarrow\quad \exists\,j\ \forall\,i,k:\ \Phi_{ik}=\Phi_{ij}\overline{\Phi_{kj}}
\quad\Rightarrow\quad \BPhi=\bvtheta\bvtheta^\herm,
    \end{align}
    where, in the last equation, $\bvtheta\in\mathbb{T}^N$ is the $j$-th column of $\BPhi$.
\end{proof}

\begin{lemma}\label[lemma]{lemma:realvalpt}
Let $\Bmrm{A} \in\C^{N\times N}$ be Hermitian and $\boldsymbol{v}\in\C^N$. For $i\neq j$, define
$a_{ij}\coloneqq\abs{\mrm{A}_{ij}}$ and $\phi_{ij}\coloneqq\angle \mrm{A}_{ij}$, and let $\varphi_{i}\coloneqq\angle v_i$.
Then it holds that
\begin{equation}
\inner*{\boldsymbol{v}\boldsymbol{v}^\herm}{\Bmrm{A}}_\frob
= \boldsymbol{v}^\herm\Bmrm{A}\boldsymbol{v}
= \sum_{i=1}^N \mrm{A}_{ii}\abs{v_i}^2
+ 2\sum_{i=1}^N\sum_{j=i+1}^N a_{ij}\abs{v_i}\abs{v_j}\cos\big(\phi_{ij}-(\varphi_i-\varphi_j)\big).
\end{equation}
\end{lemma}

\begin{proof}
Since $\Bmrm{A}$ is Hermitian, $\mrm{A}_{ii}\in\R$ for all $i$, and
\[
\boldsymbol{v}^\herm\Bmrm{A}\boldsymbol{v}
=\sum_{i=1}^N\sum_{j=1}^N \mrm{A}_{ij}\overline{v_i}v_j
=\inner*{\boldsymbol{v}\boldsymbol{v}^\herm}{\Bmrm{A}}_\frob.
\]
For $i\neq j$, write $\mrm{A}_{ij}=a_{ij}\exp(\iu\phi_{ij})$ and $v_i=\abs{v_i}\exp(\iu\varphi_i)$. Splitting diagonal and
off-diagonal terms and using $\mrm{A}_{ji}=\overline{\mrm{A}_{ij}}$, we obtain
\begin{align*}
\boldsymbol{v}^\herm\Bmrm{A}\boldsymbol{v}
&= \sum_{i=1}^N \mrm{A}_{ii}\abs{v_i}^2
+ \sum_{i=1}^N\sum_{j=i+1}^N\Big(\mrm{A}_{ij}\overline{v_i}v_j+\mrm{A}_{ji}\overline{v_j}v_i\Big)\\
&= \sum_{i=1}^N \mrm{A}_{ii}\abs{v_i}^2
+ \sum_{i=1}^N\sum_{j=i+1}^N a_{ij}\abs{v_i}\abs{v_j}
\Big(\exp\big(\iu(\phi_{ij}-(\varphi_i-\varphi_j))\big)
+\exp\big(-\iu(\phi_{ij}-(\varphi_i-\varphi_j))\big)\Big)\\
&= \sum_{i=1}^N \mrm{A}_{ii}\abs{v_i}^2
+ 2\sum_{i=1}^N\sum_{j=i+1}^N a_{ij}\abs{v_i}\abs{v_j}
\cos\big(\phi_{ij}-(\varphi_i-\varphi_j)\big).
\end{align*}
\end{proof}

\begin{proposition}\label[proposition]{proposition:lamie}
    The LaMIE estimator in~\cite{bai2023lamie} is equivalent to \Cref{eq:defPLmethods} with weight matrix $\Bmrm{W}=\abs{\BTau}\circ\abs{\BTau}$.
\end{proposition}

\begin{proof}
As in~\cite[Eq.~(13)]{bai2023lamie}, LaMIE solves
\begin{equation}\label{eq:lamie_orig}
\argmin_{\boldsymbol{v}\in\mathbb{T}^N}
\norm*{\operatorname{diag}(\boldsymbol{v})\,\abs{\BTau}\,\operatorname{diag}(\boldsymbol{v})^\herm-\BTau}_\frob^2.
\end{equation}
Set $\Bmrm{X}(\boldsymbol{v})\coloneqq \operatorname{diag}(\boldsymbol{v})\,\abs{\BTau}\,\operatorname{diag}(\boldsymbol{v})^\herm$.
Then
\begin{align}
\norm{\Bmrm{X}(\boldsymbol{v})-\BTau}_\frob^2
&= \inner*{\Bmrm{X}(\boldsymbol{v}) - \BTau}{\Bmrm{X}(\boldsymbol{v}) - \BTau}_\frob = \norm{\Bmrm{X}(\boldsymbol{v})}_\frob^2 + \norm{\BTau}_\frob^2 - 2\operatorname{Re}\inner*{\Bmrm{X}(\boldsymbol{v})}{\BTau}_\frob.\label{eq:lamie_expand_inner}
\end{align}
For $\boldsymbol{v}\in\mathbb{T}^N$, $\Bmrm{X}(\boldsymbol{v})$ is obtained from $\abs{\BTau}$ by entrywise phase factors of unit modulus, hence
$\norm{\Bmrm{X}(\boldsymbol{v})}_\frob=\norm{\abs{\BTau}}_\frob$, which is independent of $\boldsymbol{v}$.
Therefore, minimizing \eqref{eq:lamie_orig} is equivalent to maximizing $\operatorname{Re}\inner*{\Bmrm{X}(\boldsymbol{v})}{\BTau}_\frob$.
\begin{align*}
\inner*{\Bmrm{X}(\boldsymbol{v})}{\BTau}_\frob
&=\inner*{\abs{\BTau}\circ\boldsymbol{v}\boldsymbol{v}^\herm}{\abs{\BTau}\circ\BPhi}_\frob
=\inner*{\boldsymbol{v}\boldsymbol{v}^\herm}{\abs{\BTau}\circ\abs{\BTau}\circ\BPhi}_\frob=\inner*{\boldsymbol{v}\boldsymbol{v}^\herm}{\Bmrm{W}\circ\BPhi}_\frob,
\end{align*}
where $\Bmrm{W}\coloneqq \abs{\BTau}\circ\abs{\BTau}$. This quantity is real since $\Bmrm{W}\circ\BPhi$ is Hermitian (\Cref{lemma:realvalpt}); so the real-part operator can be dropped. Hence,
\[
\argmin_{\boldsymbol{v}\in\mathbb{T}^N}\norm*{\operatorname{diag}(\boldsymbol{v})\,\abs{\BTau}\,\operatorname{diag}(\boldsymbol{v})^\herm-\BTau}_\frob^2
=
\argmax_{\boldsymbol{v}\in\mathbb{T}^N}\inner*{\boldsymbol{v}\boldsymbol{v}^\herm}{\Bmrm{W}\circ\BPhi}_\frob,
\]
which is \Cref{eq:defPLmethods} with $\mathcal{X}=\mathbb{T}^N$ and $\Bmrm{W}=\abs{\BTau}\circ\abs{\BTau}$.
\end{proof}

\begin{proposition}[Inverse of the AR(1) correlation matrix]\label[proposition]{prop:ar1_inverse}
Let $N\ge 2$ and let $\zeta\in\mathbb{R}$ satisfy $0<\zeta<1$. Define the symmetric matrix
$\Bmrm{A}\in\mathbb{R}^{N\times N}$ by
\[
\forall\,i,j:\mrm{A}_{ij}=\zeta^{|i-j|}.
\]
Then $\Bmrm{A}$ has the inverse
\[
\Bmrm{A}^{-1}=\frac{1}{1-\zeta^2}\begin{pmatrix}
1 & -\zeta & 0 & \cdots & 0\\
-\zeta & 1+\zeta^2 & -\zeta & \ddots & \vdots\\
0 & -\zeta & 1+\zeta^2 & \ddots & 0\\
\vdots & \ddots & \ddots & \ddots & -\zeta\\
0 & \cdots & 0 & -\zeta & 1
\end{pmatrix}.
\]
\end{proposition}

\begin{proof}
The proposition is proven by straightforward calculation of $\Bmrm{A}^{-1}\Bmrm{A}$, given the definitions of $\Bmrm{A}$ and $\Bmrm{A}^{-1}$.
\end{proof}

\begin{lemma}\label[lemma]{lemma:eq:gofcoef01} Denote $\Pi^N\coloneqq{(-\pi,\pi]}^N$. In the context of \Cref{eq:defPTreal}, it holds that 
\begin{equation}
        \max_{\boldsymbol{\theta}\in\Pi^N}\left\{\sum_{i=1}^N\sum_{j=i+1}^N\mrm{W}_{ij}\cos\big(\phi_{ij} - (\theta_i-\theta_j)\big)\right\}\geq0.
    \end{equation}
\end{lemma}

\begin{proof}
The entries of $\Bmrm{W}$ may be negative. In the context of \Cref{eq:defPTreal}, negative weights can be interpreted as weights accompanied by a phase shift. Using $-\cos(x)=\cos(x-\pi)$, define
\begin{equation}
\tilde{\phi}_{ij}\coloneqq
\begin{cases}
\phi_{ij}, & \mrm{W}_{ij}\ge 0,\\
\phi_{ij}-\pi, & \mrm{W}_{ij}<0.
\end{cases}
\end{equation}
Then
\begin{equation}
\mrm{W}_{ij}\cos\big(\phi_{ij}-(\theta_i-\theta_j)\big)
=\abs{\mrm{W}_{ij}}\cos\big(\tilde{\phi}_{ij}-(\theta_i-\theta_j)\big).
\end{equation}
Hence, without loss of generality, we may assume $\mrm{W}_{ij}\ge 0$ for all $i,j$.

We prove the claim by induction on $N$. For the base case $N=2$, choose
$\btheta\coloneqq(\phi_{12},0)^\tran$. Then
\[
\sum_{i=1}^2\sum_{j=i+1}^2\mrm{W}_{ij}\cos\big(\phi_{ij}-(\theta_i-\theta_j)\big)
=\mrm{W}_{12}\cos(0)=\mrm{W}_{12}\ge 0.
\]

For the induction step, assume the statement holds for $N$ and consider $N+1$. Pick a maximizer \[
\hat{\btheta}\in\argmax_{\boldsymbol{\theta}\in\Pi^N}\left\{\sum_{i=1}^N\sum_{j=i+1}^N\mrm{W}_{ij}\cos\big(\phi_{ij}-(\theta_i-\theta_j)\big)\right\}
\]
and let
\[
\zeta\coloneqq \sum_{i=1}^N\sum_{j=i+1}^N\mrm{W}_{ij}\cos\big(\phi_{ij}-(\hat{\theta}_i-\hat{\theta}_j)\big)\ge 0.
\]
Define
\[
\omega_{N+1}\coloneqq
-\angle\left(\sum_{i=1}^{N}\mrm{W}_{i,N+1}\exp\big(\iu(\phi_{i,N+1}-\hat{\theta}_i)\big)\right).
\]
Then
\begin{align}
&\max_{\boldsymbol{\theta}\in\Pi^{N+1}}\left\{\sum_{i=1}^{N+1}\sum_{j=i+1}^{N+1}\mrm{W}_{ij}\cos\big(\phi_{ij}-(\theta_i-\theta_j)\big)\right\}\nonumber\\
&= \max_{\boldsymbol{\theta}\in\Pi^{N+1}}\left\{\sum_{i=1}^{N}\sum_{j=i+1}^{N}\mrm{W}_{ij}\cos\big(\phi_{ij} - (\theta_i-\theta_j)\big)+\sum_{i=1}^{N}\mrm{W}_{iN+1}\cos\big(\phi_{i,N+1} - (\theta_i-\theta_{N+1})\big)\right\}\nonumber\\
&\ge \zeta+\max_{\theta_{N+1}\in\Pi^1}\left\{\sum_{i=1}^{N}\mrm{W}_{i,N+1}\cos\big(\phi_{i,N+1}-(\hat{\theta}_i-\theta_{N+1})\big)\right\}\nonumber\\
&\ge \zeta+\sum_{i=1}^{N}\mrm{W}_{i,N+1}\cos\big(\phi_{i,N+1}-(\hat{\theta}_i-\omega_{N+1})\big)\nonumber\\
&= \zeta+\operatorname{Re}\!\left(\sum_{i=1}^{N}\mrm{W}_{i,N+1}\exp\big(\iu(\phi_{i,N+1}-\hat{\theta}_i)\big)\exp(\iu\omega_{N+1})\right)\nonumber\\
&= \zeta+\abs*{\sum_{i=1}^{N}\mrm{W}_{i,N+1}\exp\big(\iu(\phi_{i,N+1}-\hat{\theta}_i)\big)}\ \ge\ 0,
\end{align}
which completes the induction.
\end{proof}

\begin{proposition}\label[proposition]{prop:gof_bounds}
In the context of \Cref{sec:basis}, the bounds listed in \Cref{tab:gof_bounds} satisfy:
\begin{enumerate}[(a)]
\item $f_{\min}(\BTau)\le f(\BTau,\boldsymbol{\hat{\theta}})$,
\item $f(\BTau,\boldsymbol{\hat{\theta}})\le f_{\max}(\BTau)$,
\item If $\BTau=\BSigma$ and $\BSigma=\abs{\BSigma}\circ\bvtheta\bvtheta^\herm$ for some $\bvtheta\in\mathbb{T}^N$, then $f(\BSigma,\boldsymbol{\hat{\theta}})=f_{\max}(\BSigma)$.
\end{enumerate}
For PT, (c) additionally requires $\mrm{W}_{ij}\ge 0$ for all $i<j$.
\end{proposition}

\begin{proof}
We divide the proof by method category.

\begin{itemize}
    \item \textbf{PT methods:} (a) is shown in \Cref{lemma:eq:gofcoef01}. (b) follows directly from the bounds of the cosine. For (c), assume $\mrm{W}_{ij}\ge 0$ for all $i<j$ and choose $\boldsymbol{\hat{\theta}}=\angle\bvtheta$. Then $\phi_{ij}=\hat{\theta}_i - \hat{\theta}_j$ and therefore all cosine terms equal $1$. Hence
$f(\BSigma,\boldsymbol{\hat{\theta}})=\sum_{i<j}\mrm{W}_{ij}=f_{\max}(\BSigma)$, and $\boldsymbol{\hat{\theta}}$ is a maximizer.

    \item \textbf{ED methods, (a):}
In all ED cases, we have $f(\BTau,\boldsymbol{\hat{\theta}})=\lambda_{\max}(\Bmrm{A})$ for the Hermitian matrix $\Bmrm{A}=\Bmrm{W}\circ\BPhi$. The largest eigenvalue is bounded below by the average of all eigenvalues:
$f(\BTau,\boldsymbol{\hat{\theta}})\ge N^{-1}\operatorname{tr}(\Bmrm{A})$,
with equality when all eigenvalues are equal (e.g., when $\Bmrm{A}$ is a multiple of the identity).
In the EW and CW cases, $\Bmrm{A}$ has unit diagonal and thus $\operatorname{tr}(\Bmrm{A})=N$,
so the bound simplifies to $1$.
\item \textbf{CW-ED, (b) and (c):} Let $\boldsymbol{x}\in\C^N$ with $\|\boldsymbol{x}\|_2=1$, and set \[\boldsymbol{y}\coloneqq(\abs{x_1},\ldots,\abs{x_N})^\tran,\] so that $\|\boldsymbol{y}\|_2=1$.
Since $\BTau$ is Hermitian, $\boldsymbol{x}^\herm\BTau\boldsymbol{x}\in\R$ and
\[
\boldsymbol{x}^\herm\BTau\boldsymbol{x}
\le \bigl|\boldsymbol{x}^\herm\BTau\boldsymbol{x}\bigr|
=\left|\sum_{i,j}\Tau_{ij}\,\overline{x_i}x_j\right|
\le \sum_{i,j}\abs{\Tau_{ij}}\,\abs{x_i}\abs{x_j}
=\boldsymbol{y}^\tran\abs{\BTau}\,\boldsymbol{y}
\le \lambda_{\max}(\abs{\BTau}).
\]
Maximizing over $\boldsymbol{x}$ yields $\lambda_{\max}(\BTau)\le \lambda_{\max}(\abs{\BTau})=f_{\max}(\BTau)$, proving (b).
If $\BTau=\BSigma=\abs{\BSigma}\circ\bvtheta\bvtheta^\herm$, then
$\BSigma=\operatorname{diag}(\bvtheta)\,\abs{\BSigma}\,\operatorname{diag}(\bvtheta)^\herm$,
so $\BSigma$ is unitarily similar to $\abs{\BSigma}$ and thus
$\lambda_{\max}(\BSigma)=\lambda_{\max}(\abs{\BSigma})$, i.e. $f(\BSigma,\boldsymbol{\hat{\theta}})=f_{\max}(\BSigma)$.
\item \textbf{EW-ED, (b) and (c):} The argument is analogous. Since $\abs{\Phi_{ij}}=1$ for all $i,j$, we have $\abs{\BPhi}=\mathds{1}_N$, and therefore
$\lambda_{\max}(\abs{\BPhi})=\lambda_{\max}(\mathds{1}_N)=N$, which gives (b).
If $\BTau=\BSigma=\abs{\BSigma}\circ\bvtheta\bvtheta^\herm$, then $\BPhi=\bvtheta\bvtheta^\herm$ is rank one and hence
$\lambda_{\max}(\BPhi)=\|\bvtheta\|_2^2=N$, i.e. $f(\BSigma,\boldsymbol{\hat{\theta}})=f_{\max}(\BSigma)$.

\item \textbf{ML-ED, (b) and (c).}
We refer to \cite[Eq.~(19)]{ansari2018efficient}.
\end{itemize}
\end{proof}

\begin{lemma}\label[lemma]{irwinhall} For i.i.d. random variables $X_1, X_2, X_3\sim\mathcal{U}(-\pi,\pi)$, it holds that $$\mathbb{E}\sbrack*{\cos(X_1+X_2+X_3)}=0.$$
\end{lemma}

\begin{proof}
    The distribution of the sum $X_1+X_2+X_3$ is a linearly transformed Irwin-Hall distribution~\cite[p.~269]{johnson1995continuous} with density function \begin{equation}
        f_{X_1+X_2+X_3}(x)=\begin{cases}0 & x<-3\pi\quad\vee \quad x>3\pi \\ \frac{\paren*{x\pm 3\pi}^2}{16\pi^3} & x\in\sbrack*{\mp3\pi,\mp\pi} \\ \frac{-(x+\sqrt{3}\pi)(x-\sqrt{3}\pi)}{8\pi^3} & x\in\paren*{-\pi,\pi}\end{cases}
    \end{equation} and thus it follows that \begin{equation}
        \mathbb{E}\sbrack*{\cos(X_1+X_2+X_3)}=\int_{-\infty}^\infty\cos(x)f_{X_1+X_2+X_3}(x)\mrm{d}x=0.
    \end{equation}
\end{proof}

\begin{table}[ht]
    \centering
    \setlength{\tabcolsep}{6.5pt}
\renewcommand{\arraystretch}{1.3}
\begin{tabular}{@{}*{12}{c}@{}}
\toprule
\multicolumn{2}{c}{\textbf{2017}} & \multicolumn{2}{c}{\textbf{2018}} & \multicolumn{2}{c}{\textbf{2019}}
& \multicolumn{2}{c}{\textbf{2020}} & \multicolumn{2}{c}{\textbf{2021}} & \multicolumn{2}{c}{\textbf{2022}} \\
\cmidrule(lr){1-2} \cmidrule(lr){3-4} \cmidrule(lr){5-6} \cmidrule(lr){7-8} \cmidrule(lr){9-10} \cmidrule(lr){11-12}
\textbf{Date} & \boldmath$B_\perp$ & \textbf{Date} & \boldmath$B_\perp$ & \textbf{Date} & \boldmath$B_\perp$
& \textbf{Date} & \boldmath$B_\perp$ & \textbf{Date} & \boldmath$B_\perp$ & \textbf{Date} & \boldmath$B_\perp$ \\
\midrule
06-30 & -38.667 & 06-28 & 126.597 & 08-31 &  87.637 & 07-04 & 128.288 & 07-02 & -17.107 & 06-30 & 105.739 \\
07-11 & -78.317 & 07-09 & -60.286 & 09-11 &   0.000 & 07-15 & -40.667 & 07-13 &  93.047 & 07-11 &  11.488 \\
07-22 &   9.502 & 07-20 &  81.146 & 09-22 & -161.642 & 07-26 & -152.157 & 07-24 &  76.578 & 07-22 & -91.740 \\
08-02 & -205.914 & 07-31 & -60.562 &      &         & 08-06 & 120.408 & 08-04 &  79.235 & 08-02 &  42.391 \\
08-13 & -90.121  & 08-11 &  71.826 &      &         & 08-28 & -66.015 & 08-15 & -139.032 & 08-13 & -95.066 \\
08-24 &  -4.196  & 08-22 &  89.083 &      &         & 09-08 &  40.000 & 08-26 &  -61.865 & 08-24 &  37.903 \\
09-04 & -123.740 &       &         &      &         & 09-19 & -50.188 & 09-06 & -107.201 & 09-04 & -22.675 \\
09-15 & -151.986 &       &         &      &         & 09-30 &  21.200 & 09-17 &  -10.099 & 09-15 & -41.033 \\
09-26 & -199.753 &       &         &      &         &       &         & 09-28 & -140.038 & 09-26 & -56.552 \\
\bottomrule
\end{tabular}
\caption{Acquisition dates grouped by year, with perpendicular baselines $B_\perp$ relative to 2019-09-11. Dates are formatted as \texttt{MM-DD}; $B_\perp$ values are rounded to three decimal places and expressed in meters (m).}
\label{tab:acq_dates}
\end{table}

\printbibliography

@Article{berardino2002new,
  author       = {Berardino, P. and Fornaro, G. and Lanari, R. and Sansosti, E.},
  date         = {2002-11},
  journaltitle = {{IEEE} Transactions on Geoscience and Remote Sensing},
  title        = {A New Algorithm for Surface Deformation Monitoring Based on Small Baseline Differential {SAR} Interferograms},
  doi          = {10.1109/tgrs.2002.803792},
  issn         = {0196-2892},
  number       = {11},
  pages        = {2375--2383},
  volume       = {40},
  publisher    = {{IEEE}},
}

@Article{ferretti2011new,
  author       = {Ferretti, Alessandro and Fumagalli, Alfio and Novali, Fabrizio and Prati, Claudio and Rocca, Fabio and Rucci, Alessio},
  date         = {2011-09},
  journaltitle = {{IEEE} Transactions on Geoscience and Remote Sensing},
  title        = {A New Algorithm for Processing Interferometric Data-Stacks: {SqueeSAR}},
  doi          = {10.1109/tgrs.2011.2124465},
  issn         = {1558-0644},
  number       = {9},
  pages        = {3460--3470},
  volume       = {49},
  publisher    = {{IEEE}},
}

@Article{ferretti2001permanent,
  author       = {Ferretti, A. and Prati, C. and Rocca, F.},
  date         = {2001},
  journaltitle = {{IEEE} Transactions on Geoscience and Remote Sensing},
  title        = {Permanent scatterers in {SAR} interferometry},
  doi          = {10.1109/36.898661},
  issn         = {0196-2892},
  number       = {1},
  pages        = {8--20},
  volume       = {39},
  publisher    = {{IEEE}},
}

@Article{fornaro2015caesar,
  author       = {Fornaro, Gianfranco and Verde, Simona and Reale, Diego and Pauciullo, Antonio},
  date         = {2015-04},
  journaltitle = {{IEEE} Transactions on Geoscience and Remote Sensing},
  title        = {{CAESAR}: An Approach Based on Covariance Matrix Decomposition to Improve Multibaseline{\textendash}Multitemporal Interferometric {SAR} Processing},
  doi          = {10.1109/tgrs.2014.2352853},
  issn         = {1558-0644},
  number       = {4},
  pages        = {2050--2065},
  volume       = {53},
  publisher    = {{IEEE}},
}

@Article{montiguarnieri2008exploitation,
  author       = {Guarnieri, Andrea Monti and Tebaldini, Stefano},
  date         = {2008-11},
  journaltitle = {{IEEE} Transactions on Geoscience and Remote Sensing},
  title        = {On the Exploitation of Target Statistics for {SAR} Interferometry Applications},
  doi          = {10.1109/tgrs.2008.2001756},
  issn         = {0196-2892},
  number       = {11},
  pages        = {3436--3443},
  volume       = {46},
  publisher    = {{IEEE}},
}

@Article{hooper2004new,
  author       = {Hooper, Andrew and Zebker, Howard and Segall, Paul and Kampes, Bert},
  date         = {2004-12},
  journaltitle = {Geophysical Research Letters},
  title        = {A new method for measuring deformation on volcanoes and other natural terrains using {InSAR} persistent scatterers},
  doi          = {10.1029/2004gl021737},
  issn         = {1944-8007},
  number       = {23},
  volume       = {31},
  publisher    = {American Geophysical Union ({AGU})},
}

@Article{lombardini2005differential,
  author       = {Lombardini, F.},
  date         = {2005-01},
  journaltitle = {{IEEE} Transactions on Geoscience and Remote Sensing},
  title        = {Differential tomography: a new framework for {SAR} interferometry},
  doi          = {10.1109/tgrs.2004.838371},
  issn         = {0196-2892},
  number       = {1},
  pages        = {37--44},
  volume       = {43},
  publisher    = {{IEEE}},
}

@article{mora2003linear,
  author       = {O. Mora and J. J. Mallorqui and A. Broquetas},
  date         = {2003-10},
  journaltitle = {{IEEE} Transactions on Geoscience and Remote Sensing},
  title        = {Linear and nonlinear terrain deformation maps from a reduced set of interferometric {SAR} images},
  doi          = {10.1109/tgrs.2003.814657},
  publisher    = {{IEEE}}
}

@Book{press2007numerical,
  author    = {Press, W. H. and Teukolsky, S. A. and Vetterling, W. T. and Flannery, B. P.},
  title     = {Numerical Recipes},
  edition   = {3rd ed.},
  isbn      = {9780521880688},
  publisher = {Cambridge University Press},
  subtitle  = {The Art of Scientific Computing},
  year      = {2007},
}

@Article{siddique2018sar,
  author       = {Siddique, Muhammad Adnan and Wegmüller, Urs and Hajnsek, Irena and Frey, Othmar},
  date         = {2018-06},
  journaltitle = {Remote Sensing},
  title        = {{SAR} Tomography as an Add-On to {PSI}: Detection of Coherent Scatterers in the Presence of Phase Instabilities},
  doi          = {10.3390/rs10071014},
  issn         = {2072-4292},
  number       = {7},
  pages        = {1014},
  volume       = {10},
  publisher    = {{MDPI} {AG}},
}

@Article{siddique2016single,
  author       = {Siddique, Muhammad Adnan and Wegmuller, Urs and Hajnsek, Irena and Frey, Othmar},
  date         = {2016-10},
  journaltitle = {{IEEE} Transactions on Geoscience and Remote Sensing},
  title        = {Single-Look {SAR} Tomography as an Add-On to {PSI} for Improved Deformation Analysis in Urban Areas},
  doi          = {10.1109/tgrs.2016.2581261},
  issn         = {1558-0644},
  number       = {10},
  pages        = {6119--6137},
  volume       = {54},
  publisher    = {{IEEE}},
}

@InProceedings{werner2003interferometric,
  author     = {Werner, C. and Wegmuller, U. and Strozzi, T. and Wiesmann, A.},
  booktitle  = {2003 {IEEE} International Geoscience and Remote Sensing Symposium},
  title      = {Interferometric point target analysis for deformation mapping},
  doi        = {10.1109/igarss.2003.1295516},
  pages      = {4362--4364},
  publisher  = {{IEEE}},
  series     = {IGARSS-03},
  volume     = {7},
  collection = {IGARSS-03},
  year       = {2003},
}

@article{zhu2010very,
  author       = {Xiao Xiang Zhu and Richard Bamler},
  date         = {2010-12},
  journaltitle = {{IEEE} Transactions on Geoscience and Remote Sensing},
  title        = {Very High Resolution Spaceborne {SAR} Tomography in Urban Environment},
  doi          = {10.1109/tgrs.2010.2050487},
  publisher    = {{IEEE}}
}

@Article{cao2015mathematical,
  author       = {Ning Cao and Hyongki Lee and Hahn Chul Jung},
  date         = {2015-09},
  journaltitle = {{IEEE} Geoscience and Remote Sensing Letters},
  title        = {Mathematical Framework for Phase-Triangulation Algorithms in Distributed-Scatterer Interferometry},
  doi          = {10.1109/lgrs.2015.2430752},
  issn         = {1558-0571},
  number       = {9},
  pages        = {1838--1842},
  volume       = {12},
  publisher    = {{IEEE}},
}

@Article{ansari2018efficient,
  author       = {Ansari, Homa and De Zan, Francesco and Bamler, Richard},
  date         = {2018-07},
  journaltitle = {{IEEE} Transactions on Geoscience and Remote Sensing},
  title        = {Efficient Phase Estimation for Interferogram Stacks},
  doi          = {10.1109/tgrs.2018.2826045},
  issn         = {1558-0644},
  number       = {7},
  pages        = {4109--4125},
  volume       = {56},
  publisher    = {{IEEE}},
}

@Article{zwieback2022cheap,
  author       = {Zwieback, S.},
  date         = {2022},
  journaltitle = {{IEEE} Geoscience and Remote Sensing Letters},
  title        = {Cheap, Valid Regularizers for Improved Interferometric Phase Linking},
  doi          = {10.1109/lgrs.2022.3197423},
  issn         = {1558-0571},
  pages        = {1--4},
  volume       = {19},
  publisher    = {Institute of Electrical and Electronics Engineers ({IEEE})},
}

@Article{zheng2022closure,
  author       = {Yujie Zheng and Heresh Fattahi and Piyush Agram and Mark Simons and Paul Rosen},
  date         = {2022},
  journaltitle = {{IEEE} Transactions on Geoscience and Remote Sensing},
  title        = {On Closure Phase and Systematic Bias in Multilooked {SAR} Interferometry},
  doi          = {10.1109/tgrs.2022.3167648},
  publisher    = {{IEEE}},
}

@Article{frey2013dem,
  author       = {Frey, Othmar and Santoro, Maurizio and Werner, Charles L. and Wegmuller, Urs},
  date         = {2013-01},
  journaltitle = {{IEEE} Geoscience and Remote Sensing Letters},
  title        = {{DEM}-Based {SAR} Pixel-Area Estimation for Enhanced Geocoding Refinement and Radiometric Normalization},
  doi          = {10.1109/lgrs.2012.2192093},
  issn         = {1558-0571},
  number       = {1},
  pages        = {48--52},
  volume       = {10},
  publisher    = {Institute of Electrical and Electronics Engineers ({IEEE})},
}

@Article{minh2023interferometric,
  author       = {Minh, Dinh Ho Tong and Tebaldini, Stefano},
  date         = {2023-09},
  journaltitle = {IEEE Geoscience and Remote Sensing Magazine},
  title        = {Interferometric Phase Linking: Algorithm, application, and perspective},
  doi          = {10.1109/mgrs.2023.3300974},
  issn         = {2373-7468},
  number       = {3},
  pages        = {46--62},
  volume       = {11},
  publisher    = {Institute of Electrical and Electronics Engineers (IEEE)},
}

@Article{mirzaee2023non,
  author       = {Mirzaee, Sara and Amelung, Falk and Fattahi, Heresh},
  date         = {2023-02},
  journaltitle = {Computers \& Geosciences},
  title        = {Non-linear phase linking using joined distributed and persistent scatterers},
  doi          = {10.1016/j.cageo.2022.105291},
  issn         = {0098-3004},
  pages        = {105291},
  volume       = {171},
  publisher    = {Elsevier BV},
}

@Article{dong2018unified,
  author       = {Dong, Jie and Liao, Mingsheng and Zhang, Lu and Gong, Jianya},
  date         = {2018-09},
  journaltitle = {IEEE Transactions on Geoscience and Remote Sensing},
  title        = {A Unified Approach of Multitemporal {SAR} Data Filtering Through Adaptive Estimation of Complex Covariance Matrix},
  doi          = {10.1109/tgrs.2018.2813758},
  issn         = {1558-0644},
  number       = {9},
  pages        = {5320--5333},
  volume       = {56},
  publisher    = {Institute of Electrical and Electronics Engineers (IEEE)},
}

@Article{samieiesfahany2016phase,
  author       = {Samiei-Esfahany, Sami and Martins, Joana Esteves and van Leijen, Freek and Hanssen, Ramon F.},
  date         = {2016-10},
  journaltitle = {IEEE Transactions on Geoscience and Remote Sensing},
  title        = {Phase Estimation for Distributed Scatterers in {InSAR} Stacks Using Integer Least Squares Estimation},
  doi          = {10.1109/tgrs.2016.2566604},
  issn         = {1558-0644},
  number       = {10},
  pages        = {5671--5687},
  volume       = {54},
  publisher    = {Institute of Electrical and Electronics Engineers (IEEE)},
}

@Article{ansari2017sequential,
  author       = {Ansari, Homa and De Zan, Francesco and Bamler, Richard},
  date         = {2017-10},
  journaltitle = {IEEE Transactions on Geoscience and Remote Sensing},
  title        = {Sequential Estimator: Toward Efficient {InSAR} Time Series Analysis},
  doi          = {10.1109/tgrs.2017.2711037},
  issn         = {1558-0644},
  number       = {10},
  pages        = {5637--5652},
  volume       = {55},
  publisher    = {Institute of Electrical and Electronics Engineers (IEEE)},
}

@Article{yao2024phase,
  author       = {Yao, Shuyi and Balz, Timo},
  date         = {2024},
  journaltitle = {IEEE Transactions on Geoscience and Remote Sensing},
  title        = {Phase-Based Similarly Decorrelated Pixel Selection and Phase-Linking in {InSAR} Using Circular Statistics},
  doi          = {10.1109/tgrs.2024.3422171},
  issn         = {1558-0644},
  pages        = {1--12},
  volume       = {62},
  publisher    = {Institute of Electrical and Electronics Engineers (IEEE)},
}

@Article{even2018insar,
  author       = {Even, Markus and Schulz, Karsten},
  date         = {2018-05},
  journaltitle = {Remote Sensing},
  title        = {{InSAR} Deformation Analysis with Distributed Scatterers: A Review Complemented by New Advances},
  doi          = {10.3390/rs10050744},
  issn         = {2072-4292},
  number       = {5},
  pages        = {744},
  volume       = {10},
  publisher    = {MDPI AG},
}

@Book{horn2013matrix,
  author    = {Horn, R.A. and Johnson, C.R.},
  title     = {Matrix Analysis},
  edition   = {2nd ed.},
  isbn      = {9780521548236},
  location  = {New York},
  publisher = {Cambridge University Press},
  year      = {2013},
}

@Article{yu2019phase,
  author       = {Yu, Hanwen and Lan, Yang and Yuan, Zhihui and Xu, Junyi and Lee, Hyongki},
  date         = {2019-03},
  journaltitle = {IEEE Geoscience and Remote Sensing Magazine},
  title        = {Phase Unwrapping in {InSAR}: A Review},
  doi          = {10.1109/mgrs.2018.2873644},
  issn         = {2373-7468},
  number       = {1},
  pages        = {40--58},
  volume       = {7},
  publisher    = {Institute of Electrical and Electronics Engineers (IEEE)},
}

@Book{johnson1995continuous,
  author    = {Johnson, N.L. and Kotz, S. and Balakrishnan, N.},
  title     = {Continuous Univariate Distributions, Volume 2},
  edition   = {2nd ed.},
  isbn      = {9780471584940},
  location  = {New York},
  publisher = {Wiley},
  series    = {Wiley Series in Probability and Statistics},
  lccn      = {93045436},
  year      = {1995},
}

@Book{saad2007iterative,
  author    = {Saad, Youcef},
  title     = {Iterative Methods for Sparse Linear Systems},
  edition   = {2nd ed.},
  isbn      = {9780898715347},
  location  = {Philadelphia, PA},
  publisher = {Society for Industrial and Applied Mathematics},
  year      = {2007},
}

@Book{nocedal2006numerical,
  author    = {Nocedal, Jorge and Wright, Stephen J.},
  title     = {Numerical optimization},
  edition   = {2nd ed.},
  isbn      = {9780387400655},
  location  = {New York, NY},
  publisher = {Springer},
  series    = {Springer Series in Operations Research and Financial Engineering},
  year      = {2006},
}

@Article{fornaro2009four,
  author       = {Fornaro, G. and Reale, D. and Serafino, F.},
  date         = {2009-01},
  journaltitle = {IEEE Transactions on Geoscience and Remote Sensing},
  title        = {Four-Dimensional SAR Imaging for Height Estimation and Monitoring of Single and Double Scatterers},
  doi          = {10.1109/tgrs.2008.2000837},
  issn         = {1558-0644},
  number       = {1},
  pages        = {224--237},
  volume       = {47},
  publisher    = {Institute of Electrical and Electronics Engineers (IEEE)},
}

@Article{dezan2015phase,
  author       = {De Zan, Francesco and Zonno, Mariantonietta and Lopez-Dekker, Paco},
  date         = {2015-12},
  journaltitle = {IEEE Transactions on Geoscience and Remote Sensing},
  title        = {Phase Inconsistencies and Multiple Scattering in {SAR} Interferometry},
  doi          = {10.1109/tgrs.2015.2444431},
  issn         = {1558-0644},
  number       = {12},
  pages        = {6608--6616},
  volume       = {53},
  publisher    = {Institute of Electrical and Electronics Engineers (IEEE)},
}

@Article{dezan2018vegetation,
  author       = {De Zan, Francesco and Gomba, Giorgio},
  date         = {2018-11},
  journaltitle = {Remote Sensing of Environment},
  title        = {Vegetation and soil moisture inversion from SAR closure phases: First experiments and results},
  doi          = {10.1016/j.rse.2018.08.034},
  issn         = {0034-4257},
  pages        = {562--572},
  volume       = {217},
  publisher    = {Elsevier BV},
}

@InProceedings{wegmullerautomated,
  author     = {Wegmuller, U.},
  booktitle  = {IEEE 1999 International Geoscience and Remote Sensing Symposium. IGARSS’99 (Cat. No.99CH36293)},
  title      = {Automated terrain corrected SAR geocoding},
  doi        = {10.1109/igarss.1999.772070},
  pages      = {1712--1714},
  publisher  = {IEEE},
  series     = {IGARSS-99},
  volume     = {3},
  collection = {IGARSS-99},
  year       = {1999},
}

@Article{jiang2015fast,
  author       = {Jiang, Mi and Ding, Xiaoli and Hanssen, Ramon F. and Malhotra, Rakesh and Chang, Ling},
  date         = {2015-03},
  journaltitle = {IEEE Transactions on Geoscience and Remote Sensing},
  title        = {Fast Statistically Homogeneous Pixel Selection for Covariance Matrix Estimation for Multitemporal InSAR},
  doi          = {10.1109/tgrs.2014.2336237},
  issn         = {1558-0644},
  number       = {3},
  pages        = {1213--1224},
  volume       = {53},
  publisher    = {Institute of Electrical and Electronics Engineers (IEEE)},
}

@Article{vu2023robust,
  author       = {Vu, Phan Viet Hoa and Breloy, Arnaud and Brigui, Frédéric and Yan, Yajing and Ginolhac, Guillaume},
  date         = {2023},
  journaltitle = {IEEE Transactions on Geoscience and Remote Sensing},
  title        = {Robust Phase Linking in InSAR},
  doi          = {10.1109/tgrs.2023.3289338},
  issn         = {1558-0644},
  pages        = {1--11},
  volume       = {61},
  publisher    = {Institute of Electrical and Electronics Engineers (IEEE)},
}

@Article{vu2025covariance,
  author       = {Vu, Phan Viet Hoa and Breloy, Arnaud and Brigui, Frédéric and Yan, Yajing and Ginolhac, Guillaume},
  date         = {2025},
  journaltitle = {IEEE Transactions on Geoscience and Remote Sensing},
  title        = {Covariance Fitting Interferometric Phase Linking: Modular Framework and Optimization Algorithms},
  doi          = {10.1109/tgrs.2025.3550978},
  issn         = {1558-0644},
  pages        = {1--18},
  volume       = {63},
  publisher    = {Institute of Electrical and Electronics Engineers (IEEE)},
}

@Article{zebker2005accurate,
  author       = {Zebker, H.A. and Chen, K.},
  date         = {2005-04},
  journaltitle = {IEEE Geoscience and Remote Sensing Letters},
  title        = {Accurate Estimation of Correlation in InSAR Observations},
  doi          = {10.1109/lgrs.2004.842375},
  issn         = {1545-598X},
  number       = {2},
  pages        = {124--127},
  volume       = {2},
  publisher    = {Institute of Electrical and Electronics Engineers (IEEE)},
}

@Article{wang2016robust,
  author       = {Wang, Yuanyuan and Zhu, Xiao Xiang},
  date         = {2016-02},
  journaltitle = {IEEE Transactions on Geoscience and Remote Sensing},
  title        = {Robust Estimators for Multipass SAR Interferometry},
  doi          = {10.1109/tgrs.2015.2471303},
  issn         = {1558-0644},
  number       = {2},
  pages        = {968--980},
  volume       = {54},
  publisher    = {Institute of Electrical and Electronics Engineers (IEEE)},
}

@Article{bai2023lamie,
  author       = {Bai, Yusong and Kang, Jian and Ding, Xiang and Zhang, Anping and Zhang, Zhe and Yokoya, Naoto},
  date         = {2023},
  journaltitle = {IEEE Transactions on Geoscience and Remote Sensing},
  title        = {LaMIE: Large-Dimensional Multipass InSAR Phase Estimation for Distributed Scatterers},
  doi          = {10.1109/tgrs.2023.3330971},
  issn         = {1558-0644},
  pages        = {1--15},
  volume       = {61},
  publisher    = {Institute of Electrical and Electronics Engineers (IEEE)},
}

@Article{singer2011angular,
  author       = {Singer, A.},
  date         = {2011-01},
  journaltitle = {Applied and Computational Harmonic Analysis},
  title        = {Angular synchronization by eigenvectors and semidefinite programming},
  doi          = {10.1016/j.acha.2010.02.001},
  issn         = {1063-5203},
  number       = {1},
  pages        = {20--36},
  volume       = {30},
  publisher    = {Elsevier BV},
}

@Article{lopezmartinez2007coherence,
  author       = {López-Martínez, Carlos and Pottier, Eric},
  date         = {2007-02},
  journaltitle = {Applied Optics},
  title        = {Coherence estimation in synthetic aperture radar data based on speckle noise modeling},
  doi          = {10.1364/ao.46.000544},
  issn         = {1539-4522},
  number       = {4},
  pages        = {544},
  volume       = {46},
  publisher    = {Optica Publishing Group},
}

@Article{boumal2016nonconvex,
  author       = {Boumal, Nicolas},
  date         = {2016-01},
  journaltitle = {SIAM Journal on Optimization},
  title        = {Nonconvex Phase Synchronization},
  doi          = {10.1137/16m105808x},
  issn         = {1095-7189},
  number       = {4},
  pages        = {2355--2377},
  volume       = {26},
  publisher    = {Society for Industrial & Applied Mathematics (SIAM)},
}

@Article{zwieback2022reliable,
  author       = {Zwieback, Simon and Meyer, Franz J.},
  date         = {2022},
  journaltitle = {IEEE Transactions on Geoscience and Remote Sensing},
  title        = {Reliable InSAR Phase History Uncertainty Estimates},
  doi          = {10.1109/tgrs.2022.3146816},
  issn         = {1558-0644},
  pages        = {1--9},
  volume       = {60},
  publisher    = {Institute of Electrical and Electronics Engineers (IEEE)},
}

@Article{schmidt2003time,
  author       = {Schmidt, David A. and Bürgmann, Roland},
  date         = {2003-09},
  journaltitle = {Journal of Geophysical Research: Solid Earth},
  title        = {Time‐dependent land uplift and subsidence in the Santa Clara valley, California, from a large interferometric synthetic aperture radar data set},
  doi          = {10.1029/2002jb002267},
  issn         = {0148-0227},
  number       = {B9},
  volume       = {108},
  publisher    = {American Geophysical Union (AGU)},
}

@PhdThesis{adam2025improved,
  author   = {Adam, Nico Alexander},
  date     = {2025-03},
  title    = {Improved SAR Coherence Magnitude Estimates in Scenarios with Low Coherence and Small Sample Size},
  url      = {https://elib.dlr.de/215563/},
  school   = {Universit{\"a}t der Bundeswehr M{\"u}nchen},
}

@Article{zwieback2016statistical,
  author       = {Zwieback, Simon and Liu, Xingyu and Antonova, Sofia and Heim, Birgit and Bartsch, Annett and Boike, Julia and Hajnsek, Irena},
  date         = {2016-09},
  journaltitle = {IEEE Transactions on Geoscience and Remote Sensing},
  title        = {A Statistical Test of Phase Closure to Detect Influences on DInSAR Deformation Estimates Besides Displacements and Decorrelation Noise: Two Case Studies in High-Latitude Regions},
  doi          = {10.1109/tgrs.2016.2569435},
  issn         = {1558-0644},
  number       = {9},
  pages        = {5588--5601},
  volume       = {54},
  publisher    = {Institute of Electrical and Electronics Engineers (IEEE)},
}

@Article{zhang2006complex,
  author       = {Zhang, Shuzhong and Huang, Yongwei},
  date         = {2006-01},
  journaltitle = {SIAM Journal on Optimization},
  title        = {Complex Quadratic Optimization and Semidefinite Programming},
  doi          = {10.1137/04061341x},
  issn         = {1095-7189},
  number       = {3},
  pages        = {871--890},
  volume       = {16},
  publisher    = {Society for Industrial & Applied Mathematics (SIAM)},
}

\end{document}